\documentclass[pra,aps,floatfix,amsmath,superscriptaddress,twocolumn,longbibliography,nofootinbib]{revtex4-2}

\usepackage{amsfonts}%
\usepackage{amsmath}%
\makeatletter
\def\l@subsubsection#1#2{}
\makeatother
\usepackage{amssymb}%
\usepackage{graphicx}
\usepackage{color}
\usepackage{lipsum}
\usepackage{comment}
\usepackage{hyperref}
\providecommand{\U}[1]{\protect\rule{.1in}{.1in}}
\newtheorem{theorem}{Theorem}

\newtheorem{lemma}[theorem]{Lemma}

\newtheorem{proposition}[theorem]{Proposition}

\newenvironment{proof}[1][Proof]{\noindent\textbf{#1.} }{\ \rule{0.5em}{0.5em}}

\newcommand{\be} {\begin{equation}}
\newcommand{\ee} {\end{equation}}

\newcommand{\bes} {\begin{subequations}}
\newcommand{\ees} {\end{subequations}}
\newcommand{\bea} {\begin{eqnarray}}
\newcommand{\eea} {\end{eqnarray}}
\newcommand{\beq}{\begin{equation}}
\newcommand{\eeq}{\end{equation}}

\def\>{\rangle}
\def\<{\langle}
\def\Tr{\textrm{Tr}}

\renewcommand{\max}{\textrm{max}}

\newcommand{\ignore}[1]{}

\definecolor{nblue}{rgb}{0.2,0.2,0.7}
\definecolor{ngreen}{rgb}{0.2,0.6,0.2}
\definecolor{nred}{rgb}{0.7,0.2,0.2}
\definecolor{nblack}{rgb}{0,0,0}

\def\RF{\text{RF}}
\def\i{\text{in}}

\def\o{\text{out}}
\def\s{\text{S}}

\begin{document}

\title{Universal Quantum Emulator}


\author{Iman Marvian}
\affiliation{Departments of Physics, Duke University, Durham, NC 27708, USA}
\affiliation{Department of Electrical and Computer Engineering, Duke University, Durham, NC 27708, USA}
\affiliation{Research Laboratory of Electronics, Massachusetts Institute of Technology, Cambridge, MA 02139}
\author{Seth Lloyd}
\affiliation{Research Laboratory of Electronics, Massachusetts Institute of Technology, Cambridge, MA 02139}
\affiliation{Department of Mechanical Engineering, Massachusetts Institute of Technology, Cambridge, MA 02139}


\begin{abstract}
We propose a quantum algorithm that emulates the action of an unknown unitary transformation on a given input state, using multiple copies of some unknown sample input states of the unitary and their corresponding output states. The algorithm does not assume any prior information about the unitary to be emulated or the sample input states. To emulate the action of the unknown unitary, the new input state is coupled to the given sample input-output pairs in a coherent fashion. Remarkably, the runtime of the algorithm is logarithmic in $D$, the dimension of the Hilbert space, and increases polynomially with $d$, the dimension of the subspace spanned by the sample input states. Furthermore, the sample complexity of the algorithm—i.e., the total number of copies of the sample input-output pairs needed to run the algorithm—is independent of $D$ and polynomial in $d$. In contrast, the runtime and sample complexity of incoherent methods, i.e., methods that use tomography, are both linear in $D$. The algorithm is \emph{blind}, in the sense that, at the end, it does not learn anything about the given samples or the emulated unitary. This algorithm can be used as a subroutine in other algorithms, such as quantum phase estimation.

\end{abstract}

\maketitle

\section{Introduction}

A universal quantum simulator is a machine that can be programmed to mimic the dynamics of other quantum systems \cite{Lloyd:96}. The time evolution of the simulator obeys the same equations of motion as the evolution of the simulated system. A universal quantum emulator, on the other hand, is a machine that mimics the input-output relation of another system, by looking to the output of that system on some sample input states. Unlike a simulator, an emulator does not need to obey the same dynamical equations as of the emulated system.

In this paper, we introduce a quantum algorithm that emulates the action of an unknown unitary transformation (or its inverse) on unknown given input states. The algorithm couples one copy of the given input state to multiple copies of some unknown sample input-output pairs, i.e., copies of some input states of the unitary and their corresponding output states.  A priori, the algorithm  does not need to have any information about the unitary to be emulated, or the input and output states.  Naturally, the algorithm emulates the action of the unitary on states in the subspace spanned by the previously given input states, which could be much smaller than the system Hilbert space.  Indeed,  as we further explain below, the algorithm becomes particularly more efficient than naive approaches based on tomography  in regimes where the dimension of the subspace, denoted as $d$, is much smaller than the dimension of the system Hilbert space, denoted as $D$.

Obviously, having multiple copies of sample input-output pairs one can find an approximate classical description of these states in a standard basis via state tomography. Then, based on this information, one can determine the classical description of a unitary transformation that transforms the input states to the corresponding output states. However, for large Hilbert spaces this approach is highly inefficient and impractical: First of all, state tomography requires lots of copies of the sample states. Second, to find a unitary that relates the input and output states, one needs to solve a set of linear equations with a large number of variables. And finally, even if one finds an example of such unitary, in general, this unitary cannot be implemented efficiently. 

More precisely, to emulate the action of an unknown unitary transformation (or its inverse) using the approaches based on tomography the runtime and the sample complexity, that is the total number of copies of the samples needed to run the algorithm, are both lower bounded by $\Omega(D\times d)$, where $D$ is the dimension of the system Hilbert space, and $d$ is the dimension of the subspace spanned by the sample input states. In contrast, the runtime of the algorithm proposed in this work is $\mathcal{O}(\log D\times \text{poly}(d))$, where poly($d$) denotes a polynomial in $d$.  Furthermore, the sample complexity of this algorithm is independent of $D$ and poly($d$). Therefore, under the practical assumption that the unitary transformation should be emulated in a subspace of the Hilbert space, whose dimension $d$ is finite or poly-logarithmic in the dimension of the Hilbert space, our algorithm has exponentially better runtime and sample complexity.

It is interesting to compare this result with the scenario studied in \cite{bisio2010optimal}, where one wants to \emph{learn} an unknown unitary $U$ by applying it for a finite number of times to some quantum states, so that later, when one does not have access to $U$, one can reproduce its effect on a new input state. It turns out that the strategy that maximizes the average fidelity, where the average is taken over all states in the Hilbert space, is an incoherent measure-and-rotate strategy, i.e., a method that uses tomography \cite{bisio2010optimal}. In contrast to this result, our work shows that under the practical assumption that the action of the unitary should be emulated in a low-dimensional subspace, and not the entire Hilbert space, the coherent methods are much more powerful than the incoherent ones.

The main working principle behind our algorithm is simple and intuitive: It first employs a novel technique to find the coordinates of the new input state relative to a frame defined by the sample input states in a randomized fashion and encodes these coordinates in some ancillary qubits. Next, it reconstructs the state that has exactly the same coordinates relative to the frame defined by the output sample states in a similar way. Therefore, instead of using tomography, which aims to find the description of the given sample quantum states relative to some background (classical) reference frame, we use the given samples themselves as quantum reference frames \cite{QRF_BRS_07, MS11, Modes, gour2008resource, marvian2008building}. The superiority of this method over approaches based on tomography is essentially another manifestation of the general principle that using a random background reference frame can make a problem much more complicated, and often, a simpler description of physical phenomena can be found in terms of the relational degrees of freedom.   


\section{Preliminaries}

In the following, 
 $$S_\i=\big\{\rho_k^\i\ \ : k=1,\cdots,K\big\}$$
 denotes the set of sample input states of an unknown unitary $U$, and 
 $$S_\o=\{\rho_k^\o=U \rho_k^\i U^\dag\ \ : k=1,\cdots,K\}$$
    are the corresponding output states.  We first present the algorithm for the special case of pure sample states,  where 
$$\rho_k^\i=|\phi^\i_k\rangle\langle\phi_k^\i| \ \ \ \  : k=1,\cdots,K\ ,$$    
 and
$$\rho_k^\o=|\phi^\o_k\rangle\langle\phi_k^\o|=U  |\phi^\i_k\rangle\langle\phi_k^\i|U^\dag\ : k=1,\cdots,K\ .$$    
Later, we  explain how the algorithm can  be generalized to the case of mixed states as well.  Let 
    \be
    \mathcal{H}_\i=\text{Span}_\mathbb{C}\{|\phi^\i_k\rangle : k=1,\cdots,K\}\ ,
    \ee 
and  $\mathcal{H}_\o=U\mathcal{H}_\i,$  be the subspaces  spanned by  the input and output vectors, respectively, and 
    \be
d=\dim(\mathcal{\mathcal{H}_\i})=\dim(\mathcal{H}_\o)\ , 
\ee   
 be their dimension. When the input states are mixed, $\mathcal{H}_\i$ can be defined as the subspace spanned by the union of the supports of ${\rho_k^\i}$, and $\mathcal{H}_\o = U\mathcal{H}_\i$.  As we explain later, the algorithm requires multiple copies of each sample state in $S_\i$  and $S_\o$ (Interestingly, at the end of the algorithm most of these states remain almost unchanged!).

    
    \subsection{The necessary and sufficient condition for emulation}
    
 To be able to emulate the action of an unknown  unitary transformation   $U$ on the input subspace $\mathcal{H}_\i$,  we naturally  need to assume that the  set of input states $S_\i$ uniquely determines the action of $U$ on this subspace, up to a global phase.
 To understand this requirement consider the following example:  Suppose the goal is to emulate the action of a qubit unitary $U=e^{i \alpha Z}$, which rotates the state of qubit around the z-axis by an unknown angle $\alpha\in[0,2\pi)$.  If this  unitary is applied on the sample input states $|0\rangle$ and $|1\rangle$, then the corresponding output states $e^{i\alpha}|0\rangle$ and $e^{-i\alpha}|1\rangle$ do not contain any information about the unknown angle $\alpha$, because $\alpha$ only shows up  as a global phase, which is not observable.  Therefore, even though the sample input states span the entire Hilbert space, it is impossible to emulate the action of this unknown unitary.  Of course, in this example, the issue can be easily rectified by using a non-orthogonal basis as the sample input states. In the following, we formulate this property more generally.  

Suppose a pair of unitaries  $U$ and $V$ satisfy
         \begin{equation}\label{cond1}
U \rho_k^\i U^\dag=V \rho_k^\i V^\dag \ \ \   :  \forall \rho_k^\i \in S_\i \  ,
  \end{equation} 
  which means their output density operators are identical for all input density operators in   $S_\i$. Then, we require that they should  also satisfy 
  \begin{equation}\label{cond2}
U|\psi\rangle\langle\psi|U^\dag= 
  V|\psi\rangle\langle\psi|V^\dag \  \ \  :\  \forall |\psi\rangle\in  \mathcal{H}_\i \    , 
\end{equation} 
  which means, up to a possible global phase, $U$ and $V$ act identically on all input states in $\mathcal{H}_\i $.  Eq.(\ref{cond1}) implies  that the unitary $V^\dag U$  commutes with all elements of  $S_\i$.  Eq.(\ref{cond2}), on the other hand,  means that $V^\dag U$  commutes  with $\mathcal{L}(\mathcal{H}_\i)$, the set of linear operators with support  restricted to $\mathcal{H}_\i$. This  means $V^\dag U$ acts as a global phase on $\mathcal{H}_\i$.

 It follows that  Eq.(\ref{cond1}) implies Eq.(\ref{cond2})  if, and only if
\begin{equation}\label{comm}
\text{Comm}(S_\i)=\text{Comm}(\mathcal{L}(\mathcal{H}_\i))\ ,
\end{equation}
or, equivalently  if, and only if, 
  \begin{equation}\label{alg}
  \text{Alg}_\mathbb{C}(S_\i)=\mathcal{L}(\mathcal{H}_\i)\ ,
  \end{equation}
  where $\text{Comm}(\cdot)$ denotes, the commtuant of a set of operators, and $ \text{Alg}_\mathbb{C}(S_\i)$ denotes the 
 complex associative algebra generated by  $S_\i$, i.e., the set of polynomials generated by elements of  $S_\i$ with complex coefficients.   The equivalence of Eq.(\ref{comm}) and Eq.(\ref{alg}) can be seen, e.g.,  using the bicommutant theorem, or, using the characterization of the finite-dimensional complex associative algebras \cite{Kribs:2005:180501}.  Note that in this argument states $\rho_k^\i$ do not need to be pure.  In summary,  
\begin{proposition}\label{prop1}
Let $\mathcal{H}_\i$ be the subspace spanned by the supports of (possibly mixed) density operators $S_\i=\{\rho_k^\i: k=1,\cdots K\}$. Then,  equations 
\be
\rho_k^\o=U \rho_k^\i U^\dag\ \ \ \  : k=1,\cdots K \ ,
\ee
 uniquely determine the action of an unknown unitary $U$ on this subspace (up to a possible global phase) if, and only if, $\text{Alg}_\mathbb{C}(S_\i)=\mathcal{L}(\mathcal{H}_\i)$.     
\end{proposition}
 Therefore, in the following, we naturally assume that this condition is satisfied.   As an example, suppose  $S_\i$ contains $d=\text{dim}(\mathcal{H}_\i)$ linearly independent  pairwise non-orthogonal pure states $\{|\phi^\i_k\rangle\}$, such that 
\be\label{non}
\langle\phi^\i_j|\phi^\i_k\rangle\neq 0\ \ \ \ \ \ \ \ :\ j,k=1,\cdots, d\ .
\ee
Then, for any $ j,k=1,\cdots, d$, it holds that
$$|\phi^\i_j\rangle\langle\phi^\i_k|=  \frac{|\phi^\i_j\rangle\langle\phi^\i_j|\phi^\i_k\rangle\langle\phi^\i_k|}{\langle\phi^\i_j|\phi^\i_k\rangle}\in \text{Alg}_\mathbb{C}(S_\i)\ ,$$
which, in turn, implies this algebra is equal to $\mathcal{L}(\mathcal{H}_\i)$.\footnote{
Note that Eq.(\ref{non}) is not necessary for the algorithm to work. We only need the weaker assumption in Eq.(\ref{alg}).}

\begin{figure}[t]
  \includegraphics[scale=1.25]{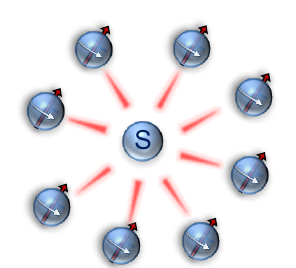}
  \caption{\textbf{A Physical Interpretation of Density Matrix Exponentiation:} Consider a spin-half system $S$ coupled to a set of $n$ ancillary spin-half systems via the Heisenberg exchange interaction, with the total Hamiltonian $H=J \sum_i \vec{\sigma}_S\cdot \vec{\sigma}_i$.  Suppose each ancillary system is initially prepared in the pure state $|\phi\rangle$. Then, as we turn on the interaction, the system S starts rotating around the axis defined by the Bloch vector of state  $|\phi\rangle$. More precisely, for $n\gg 1$, the system’s evolution can be approximated by the Hamiltonian  $(n+1)J  |\phi\rangle\langle\phi|$ (See Appendix \ref{App:DME2}). Indeed, this can be thought of as  a simple model for how a (classical) magnet interacts with a spin half-system. Note that each Heisenberg interaction   $J \vec{\sigma}_S\cdot \vec{\sigma}_i$  can be written as a linear combination of SWAP and the identity operator. Therefore, exactly the same dynamics is observed if we couple the system S to the ancillary spin-half systems via the SWAP Hamiltonian. A similar result can be established for systems with Hilbert space of arbitrary dimension. Density matrix exponentiation technique proposed in \cite{lloyd2014quantum} can be thought of as the Trotterized version of this evolution, where the system S sequentially interacts with each of these $n$ ancillary systems.}
  \label{FigDen}
\end{figure}

\begin{figure*}
\includegraphics[width=\textwidth,height=6.8cm]{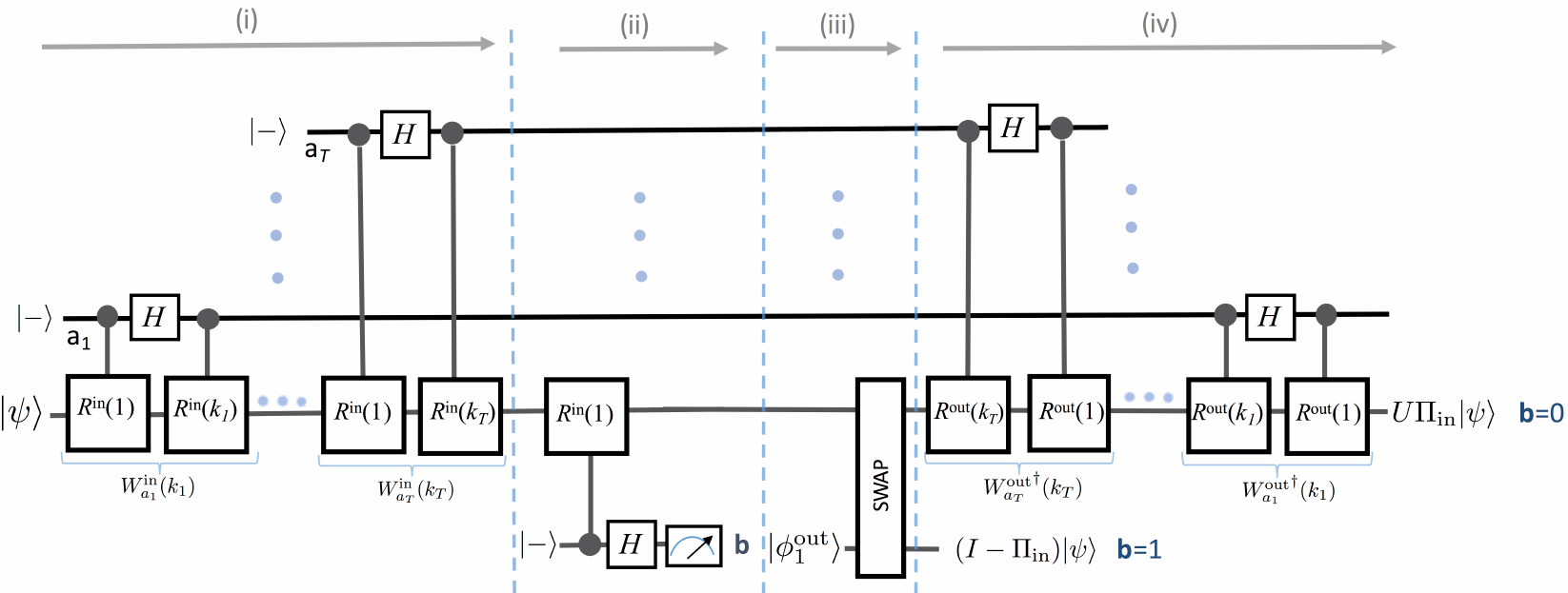}
\caption{The quantum circuit for emulating  unitary transformation $U$ for the special case of pure input-output sample pairs. Here $k_1,\cdots , k_T$ are $T=\text{poly}(d)$ integers chosen uniformly at  random from integers  $2,\cdots ,K$.  We use the given  copies of sample states in $S_\i$ and $S_\o$  to approximately implement the controlled-reflections $R_a^\i(k)$ and $R_a^\o(k)$, respectively. } 
\label{Fig}
\end{figure*}


\subsection{Controlled reflections}

 The main use of the given copies of sample states $S_\i$ and $S_\o$ is to simulate \emph{controlled-reflections} about these states.  Let 
\be
R(k)= \exp({i \pi |\phi_k\rangle\langle\phi_k|})=I-2|\phi_k\rangle\langle\phi_k|\  \ ,
\ee
 be the reflections about  state $|\phi_k\rangle$, where   we have suppressed the superscripts \emph{in} and \emph{out} in both sides. In the following algorithm, we need to implement the controlled-reflections  $R_a^\i(k)$ and $R_a^\o(k)$, defined as 
\begin{align}\label{eq:cont}
R_a(k)&=   |0\rangle\langle 0|_a \otimes I +   |1\rangle\langle 1|_a \otimes \exp({i \pi |\phi_k\rangle\langle\phi_k|})\ ,
\end{align}
where $a$ labels the control qubit,  $I$ is the identity operator on the main system, and  again we have suppressed the superscripts \emph{in} and \emph{out} on both sides.   

Using the given  copies of the sample states, we can efficiently simulate these controlled-reflections via the \emph{density matrix exponentiation} technique of \cite{lloyd2014quantum} (See also \cite{kjaergaard2022demonstration}  for an experimental demonstration of this technique). 
   In Fig.\ref{FigDen} and Appendix \ref{App:DME2} we explain a simple physical interpretation of this technique, as  
   simulating the Heisenberg interaction between the main system and the copies of sample states. Applying this technique and using $n$ copies of an unknown state $\sigma$, one can implement the unitary $\exp(-i t\sigma)$ or its controlled version $|0\rangle\langle 0| \otimes I + |1\rangle\langle 1| \otimes \exp(-i t\sigma)$ for any real $t$, with an error of $\epsilon=\mathcal{O}\left(\frac{t^2}{n}\right)$, as quantified in terms of the trace distance \cite{lloyd2014quantum}  (See also \cite{kimmel2017hamiltonian} for an error analysis and the proof of optimality of this technique).   As we further explain  in Appendix \ref{app:SWAP},  using a method previously employed in \cite{pichler2016measurement}, this unitary  can be realized with a circuit that uses only $\Theta(n\times \log D)$ 
 single-qubit gates and  controlled-controlled-SWAP gates.   

  

Therefore, in the following, where we present the algorithm, we  assume  all the controlled-reflections $\{R_a(k):  1\le k\le K\}$ can be  efficiently implemented.  Finally, note that for some applications of the emulator algorithm,  one may have the classical description of some of the sample input-output states. In this case 
the controlled-reflections can be implemented directly, without using density matrix exponentiation.\\

\noindent\textbf{Notation:} In the following  $H_a$ denotes the Hadamard gate $H$ acting on ancilla qubit $a$,  which is defined by $H|0\rangle=|+\rangle$ and $H|1\rangle=|-\rangle$ with $|\pm\rangle=(|0\rangle\pm|1\rangle)/\sqrt{2}$. To simplify the presentation, we use the notation 
\be
W_{a}(k)\equiv  R_{a}(k) H_a R_a(1)\ ,
\ee 
which describes the composition of two controlled-reflections $R_a(1)$ and $R_{a}(k)$ with a Hadamard $H_a$  gate in between them  (See Fig.(\ref{Fig})).  Note that we have again suppressed \emph{in} and \emph{out} superscripts on both sides.  The algorithm also uses a SWAP gate defined by $\text{SWAP} |\nu\rangle|\mu\rangle=|\mu\rangle|\nu\rangle$, for any pair of states $|\mu\rangle$ and $ |\nu\rangle$.


\section{The algorithm (Special case)}

In this section we discuss the algorithm for the universal quantum emulator, in the special case where all the sample input-output pairs are pure states.  Fig.(\ref{Fig}) presents the quantum circuit that emulates the action of an unknown unitary transformation $U$ on any given  state $|\psi\rangle$ in the input subspace $\mathcal{H}_\i$.  For a general input state, which is not restricted to this subspace, this circuit first projects the state to this subspace, and if successful, then applies the unitary $U$ to it. 

In this algorithm $(k_1,\cdots , k_T)$ are $T$ integers chosen uniformly at  random from integers  $2,\cdots,K$, where $T$ is a constant that determines the precision of the algorithm, and we choose it to be polynomial in $d$ the dimension of the subspace $\mathcal{H}_\i$, and independent of $D$.  Furthermore, state $|\phi_1^\i\rangle$ (and $|\phi_1^\o\rangle$) is one of the sample input states (and its corresponding output) which is chosen  randomly at the beginning of the algorithm, and is fixed during the algorithm.  In steps (i) and (iv) of the algorithm we implement, respectively,  the unitaries $W_{a_i}^{\i}(k_i)$ and ${W_{a_i}^{\o}}^\dag(k_i)$ on the system and qubit $a_i$, for $i=1,\cdots,T$.  As we explained before, all the conditional reflections $ R^\i_{a}(k) $ and $ R^\o_{a}(k) $ can be  efficiently simulated using the given copies of states $|\phi_k^\i\rangle$ and $|\phi_k^\o\rangle$. 

 In step (ii) of the algorithm we perform a qubit measurement in the computational  basis  $\{|0\rangle,|1\rangle\}$.  This measurement allows us to determine if  the initial state $|\psi\rangle$ is in the input subspace  
 $\mathcal{H}_\i$ or not.  
In particular,   with  probability approximately  $1-\langle\psi|\Pi_\i|\psi\rangle$ the qubit is projected to state $|1\rangle$, corresponding to  outcome $b=1$,  in which case  the system will be projected to a state close to $(I-\Pi_\i)|\psi\rangle/\sqrt{1- \langle\psi|\Pi_\i|\psi\rangle}$, where $\Pi_\i$ is the projector to the subspace $\mathcal{H}_\i$.   
 On the other hand, with probability  approximately equal to $\langle\psi|\Pi_\i|\psi\rangle$ we obtain the outcome $b=0$, in which case the output state of the algorithm is close to $U \Pi_\i |\psi\rangle/\sqrt{\langle\psi|\Pi_\i|\psi\rangle}$. In this case the algorithm consumes a copy of state $|\phi_1^\o\rangle$, and  returns a copy of state $|\phi_1^\i\rangle$.

A few remarks are in order: 
\begin{itemize}

\item As we show below, although the algorithm uses random integers $(k_1,\cdots , k_T)$, for sufficiently large $T$, it \emph{always} transforms the input state $|\psi\rangle\in\mathcal{H}_\i$ to a state with high fidelity to the desired output state $U|\psi\rangle$.

\item By reversing the role of the input and output sample states in the above circuit, we can implement the inverse unitary $U^\dag$ instead of $U$, using the same set of sample states.

\item As we explain in Sec. \ref{Sec:comp}, and Appendix \ref{App:Sec:comp}, a modified version of this circuit can be implemented using only $\lceil \log (T+1)\rceil$ ancillary qubits (instead of $T$ qubits), which achieves exactly the same fidelity. 

\item In general, the sample input and output states may live in different Hilbert spaces. As long as the Gram matrices of the two sets are equal, such that $\langle\phi^\i_j |\phi^\i_k\rangle=\langle\phi^\o_j |\phi^\o_k\rangle$ for $j,k=1,\cdots, K$, the algorithm  realizes the corresponding unitary  transformation from the input subspace $\mathcal{H}_\i$ to the output subspace $\mathcal{H}_\o$.         


\end{itemize}

\subsection{How it works: Coherent erasing}\label{Sec:alg}
We first assume the initial state $|\psi\rangle$ is in the input subspace $\mathcal{H}_\i$, and then  consider the general case.  To understand step (i) of the algorithm, it is useful to  focus on the reduced state of the system.  As explained below, from this  point of view, during step (i) we are trying to erase the initial state of the  system and push it into state $|\phi_1^\i\rangle$,  a state which is  chosen randomly from the sample input set $S_\i$. This is achieved via repetitive \emph{measuring and conditional mixing}. 

Consider the first round of step (i), where we randomly choose integer $k_1$ in $2,\cdots,K$, and apply the unitary $W(k_1)=R_{a_1}(k_1) H_{a_1} R_{a_1}(1)$ to the system and qubit $a_1$ initially prepared in state $|-\rangle$.  The effect of this transformation on the reduced state of the system can be interpreted in the following way: we perform the projective measurement  with projectors $$ P=|\phi_1^\i\rangle\langle \phi_1^\i|\ ,\ \ \ \   P^\perp=I-P.$$ 
Then, if the system is found in state $|\phi_1^\i\rangle$, we leave it unchanged;  otherwise, we apply the random reflection $R^\i(k_1)$ to it. Tracing over  qubit $a_1$ and averaging over the values of $k_1$, we find that the overall effect of this transformation on the reduced state of the system can be described by the quantum channel 
\be
\mathcal{W}(\rho)=P\rho P + \mathcal{D}(P^\perp \rho P^\perp)\ ,
\ee
 where
\beq
 \mathcal{D}(\rho)=\frac{1}{K-1} \sum^K_{k=2} R^\i(k)\ \rho\ R^\i(k)\ . 
 \eeq
In step (i) we repeat the above process  $T$ times with $T$ different ancillary qubits and  random integers. Using the facts that (1) ancillary qubits are initially uncorrelated with each other, and (2) different random integers $k_1,\cdots , k_T$ are statistically independent  of each other, we find that  at the end of step (i) the average reduced state of the system is described by state $\mathcal{W}^T(|\psi\rangle\langle\psi|)$. 


Next, recall that, as stated in Eq.(\ref{alg}),   the input set $S_\i$ generates the full matrix algebra in $\mathcal{H}_\i$, an assumption that is crucial for being able to uniquely determine the action of $U$ on $\mathcal{H}_\i$. Remarkably, this assumption now translates to the fact that inside $\mathcal{H}_\i$,
channel $\mathcal{W}$ has a unique fixed point state, namely $|\phi_1^\i\rangle\langle\phi_1^\i|$  and $\mathcal{W}^T(|\psi\rangle\langle\psi|)$ converges exponentially fast to this state.

\begin{proposition}\label{prop2}
If $\text{Alg}_\mathbb{C}(S_\i)=\mathcal{L}(\mathcal{H}_\i)$, then for any state $\rho$ with support restricted to  $\mathcal{H}_\i$\be\label{limit}
\lim_{T\rightarrow \infty} \mathcal{W}^T(\rho)=|\phi_1^\i\rangle\langle\phi_1^\i|\ .
\ee
\end{proposition}
\begin{proof}
Since the reflections $R^\i(k)$ are block-diagonal with respect to the input subspace $\mathcal{H}_\i$, under channels $\mathcal{D}$ and $\mathcal{W}$ any initial state $\sigma$  with support restricted to $\mathcal{H}_\i$ remains inside this subspace. Furthermore, it can be easily seen that  for any  integer $t\ge 1$, 
\be\label{mono}
\langle\phi_1^\i|\mathcal{W}^t(\sigma)|\phi_1^\i\rangle \ge \langle\phi_1^\i|\sigma|\phi_1^\i\rangle\ ,
\ee
which means the probability of finding the system  in state $|\phi_1^\i\rangle$ is non-decreasing with $t$.   Indeed,  Eq.(\ref{mono}) holds as equality only when for all $j=1,\cdots, t$,  the operator $\mathcal{D}^{j}(P^\perp \sigma P^\perp)$ remains orthogonal to $|\phi_1^\i\rangle$, in which case
  $$\mathcal{W}^t(\sigma)=P\sigma P+ \mathcal{D}^t(P^\perp \sigma P^\perp)\ ,$$
and the difference between the right- and left-hand sides of   Eq.(\ref{mono}) becomes  
 $\langle \phi_1^\i| \mathcal{D}^t(P^\perp \sigma P^\perp)|\phi_1^\i\rangle$.   
 In the following, we show that
the assumption $\text{Alg}_\mathbb{C}(S_\i)=\mathcal{L}(\mathcal{H}_\i)$  implies that for any non-zero  positive semi-definite operator $A$ on $\mathcal{H}_\i$,
\be\label{boundg}
\langle \phi_1^\i|\mathcal{D}^{t}(A)|\phi_1^\i\rangle=\Tr\big(A \mathcal{D}^{t}(|\phi_1^\i\rangle\langle \phi_1^\i|)\big) >0\ ,
\ee
  for some integer $t\in\{0, 1,\cdots,d-1\} $.  Applying this to the operator  $A=P^\perp \sigma P^\perp$ we find that, unless this operator is zero, which happens only when $\sigma=|\phi_1^\i\rangle\langle \phi_1^\i|$, 
$  \langle \phi_1^\i|\mathcal{D}^t(P^\perp \sigma P^\perp) |\phi_1^\i\rangle >0$,  and therefore Eq.(\ref{mono}) is a strict inequality for some $t\le d-1$. In conclusion,  the fidelity $f(T)=\langle\phi_1^\i|\mathcal{W}^T(\rho)|\phi_1^\i\rangle$  is non-decreasing with $T$ and unless $f(T)=1$,  $f(T+d-1)>f(T)$. Since $f(T)\le 1$, by the monotone convergence theorem, we arrive at Eq.(\ref{limit}).

To complete the proof, we prove the bound  in Eq.(\ref{boundg}) holds for some integer $t\in\{0, 1,\cdots,d-1\} $.  First, note that the support of  
 the positive semi-definite operator $\mathcal{D}^{j}(|\phi_1^\i\rangle\langle \phi_1^\i|)$ is  equal to   the span of the set of vectors 
$$V_j=\{ R^\i(k_j)\cdots R^\i(k_1)|\phi_1^\i\rangle:  k_1,\cdots, k_j=2,\cdots, K \}\ ,$$
with $V_0=\{|\phi_1^\i\rangle\}$.  The dimension of the span of the union of these sets, that is the span of   $\cup_{j=0}^t V_j$,  monotonically increases with $t$ until it reaches $d=\text{dim}(\mathcal{H}_\i)$, because otherwise we have found a proper subspace of $\mathcal{H}_\i$
which is invariant under the action of $\{R^\i(k): k=1,\cdots, K\}$,  in contradiction with  the assumption $\text{Alg}_\mathbb{C}(S_\i)=\mathcal{L}(\mathcal{H}_\i)$.  Then, since by assumption $A$ is a non-zero positive semi-deifnite operator with support restricted to $\mathcal{H}_\i$, $\langle\eta|A|\eta\rangle>0$ for a vector in $|\eta\rangle \in\cup_{j=0}^{d-1} V_j$. We conclude that  $\Tr\big(A \mathcal{D}^{t}(|\phi_1^\i\rangle\langle \phi_1^\i|)\big)> 0$ 
for some integer $t\in\{0, 1,\cdots,d-1\} $. This completes the proof of the proposition.     \end{proof}\\

To summarize, in step (i)  we \emph{erase} the initial state of the system and push it into state $|\phi_1^\i\rangle$. The fact that we have enough resources to   uniquely determine the action of unitary $U$ on any state in the input subspace $\mathcal{H}_\i$, translates to the fact that any state in this subspace can be erased in a coherent unitary fashion. 

 But, recall that the entire process is performed via unitary transformations and quantum information is conserved during any unitary evolution. This means that all the information about state $|\psi\rangle$ should now be encoded in the ancillary qubits $\textbf{a}=a_1\cdots a_T$. Furthermore, since we start with a global pure state, and since at the end of step (i)   the reduced state of the main system is close to a pure state, namely $|\phi_1^\i\rangle$,  we find that the joint reduced state of ancillary qubits $\textbf{a}$ should also be close to a pure state,  denoted by $|\Psi(\textbf{k})\rangle_\textbf{a}$. 
Note that in addition to state $|\psi\rangle$,  state $|\Psi(\textbf{k})\rangle_\textbf{a}$  also depends on the sample set $S_\i$, and the random integers $\textbf{k}\equiv (k_1,\cdots,k_T)$. 
We conclude that at the end of step (i) with high probability the system and ancillary qubits $\textbf{a}$ are in the product state 
\beq\label{EqAnc}
\left[{W_{a_{T}}^\i}({k_{T}})\cdots  {W_{a_{1}}^\i}({k_{1}}) \right]|\psi\rangle |-\rangle^{\otimes T} \approx |\phi_1^\i\rangle |\Psi(\bf{k})\rangle_\textbf{a}  \ .
\eeq
Next, in step (ii) we perform a controlled reflection followed by a single-qubit measurement, whose effect on the system is equivalent to the measurement with projectors $P=|\phi_1^\i\rangle\langle \phi_1^\i|, P^\perp=I-P$.  For any initial state $|\psi\rangle\in\mathcal{H}_\i$,
 at this point we know that  with probability almost 1 the system should be in state $|\phi_1^\i\rangle$, in which case the measurement 
projects the ancillary qubit to state $|0\rangle$. On the other hand, if we project the qubit to state $|1\rangle$ it means the erasing has not been successful.

Finally, assuming $|\psi\rangle\in\mathcal{H}_\i$, we find that applying steps (iii) and (iv) maps the system  to state $U |\psi\rangle$. This can be seen by multiplying both sides of Eq.(\ref{EqAnc}) in unitary $U$,  and using  the facts that $U|\phi_1^\i\rangle=|\phi_1^\o\rangle$,  and 
\be
(U\otimes I_a)  W^\i_a(k) (U^\dag\otimes I_a)=W^\o_a(k)\ .
\ee
 This implies
\beq\label{enc}
\left[{W_{a_{1}}^\o}^\dag({k_{1}})\cdots  {W_{a_{T}}^\o}^\dag({k_{T}}) \right] |\phi_1^\o\rangle |\Psi(\textbf{k})\rangle_\textbf{a} \approx U|\psi\rangle |-\rangle^{\otimes T}\ .
\eeq
Eq.(\ref{enc}) means that, after step (ii) by preparing the system in the output sample state $|\phi_1^\o\rangle$, which is given to us, and running  all the operations in step (i) backward, with  unitaries $R_{a}^\i(k)$ replaced by $R_{a}^\o(k)$, we obtain the system in the desired  state $U|\psi\rangle$ and ancilla qubits go back to their initial state $|-\rangle^{\otimes T}$.
 
Using the fact that all unitaries $W_{a}(k)=R_{a}(k) H_a R_a(1)$ act trivially on the subspace orthogonal to $\mathcal{H}_\i$, it can be easily seen that for a general input $|\psi\rangle$, which is not contained in $\mathcal{H}_\i$, the algorithm first  performs a projective measurement that projects the input state into the subspace $\mathcal{H}_\i$, or the orthogonal subspace. Then,  if the system is found to be in $\mathcal{H}_\i$, which  corresponds to outcome $b=0$ in the measurement in step (ii), it applies the unitary $U$ to the component of state in this subspace.

In Appendix  \ref{App:1}  we prove the following quantitative version of the above result: Suppose we implement the quantum circuit presented in Fig. \ref{Fig}, using perfect controlled-reflections $R_a(k)$. Let $\mathcal{E}_{U}$ be the quantum channel describing  the overall effect of the circuit on the system, in the case where  we do not post-select based on the outcome of the measurement in step (ii), which means we do not care if erasing has been successful or not. Then, for an arbitrary input state  $\rho$ (not necessarily restricted to subspace $\mathcal{H}_\i$), the fidelity  of $\mathcal{E}_{U}(\rho)$ and the desired state $U\rho U^\dag$ satisfies
\beq\label{main_eq}
F(\mathcal{E}_{U}(\rho), U\rho U^\dag)\ge p_\text{erase}(\rho) = \langle\phi^\i_1|\mathcal{W}^T(\rho)|\phi^\i_1\rangle\ .
\eeq
Here,  $ p_\text{erase}(\rho)$ is the probability  that we have successfully erased the state of system (and pushed it into state $|\phi^\i_1\rangle$), which corresponds to outcome $b=0$ in the measurement, and fidelity $F(\sigma_1,\sigma_2)=\|\sqrt{\sigma_1}\sqrt{\sigma_2}\|_1$, where  $\|\cdot\|_1$ denotes the $l_1$ norm, i.e., the sum of singular values\footnote{Note that here we follow the definition of fidelity in \cite{nielsen2000quantum}. Some authors call this quantity "the square root of fidelity".} . (Indeed, we show a stronger result: The fidelity of the joint state of the main system and the ancilla qubits  with the ideal desired state $U\rho U^\dag \otimes |-\rangle\langle -|^{\otimes T}$ is lower bounded by $p_\text{erase}(\rho)$). On the other hand, if we postselect to the cases where the erasing has been successful, then for pure input state $\rho$, the fidelity between the output of the algorithm and the desired state $U\rho U^\dag$ is lower bounded by $\sqrt{p_\text{erase}(\rho)}$.

Eq.(\ref{main_eq}) best captures the working principle behind this algorithm, which can  be called \emph{emulating via coherent erasing}. 
Note that using this equation we can determine for which input states, the emulation works better: if we have the required resources to coherently erase state $\rho$ and bring the system to a pure state which we know how transforms under unitary $U$, then we can emulate the action of unitary $U$ on $\rho$. 

It is also worth mentioning that this algorithm is 
distortion free, in the following sense: For any unitary $U$, the overall channel $\mathcal{E}_U$ can be decomposed as 
\be
\mathcal{E}_U=\mathcal{U}\circ \mathcal{E}_I\ , 
\ee  
where $\mathcal{U}(\cdot)=U(\cdot)U^\dag$ is the ideal channel corresponding to the unitary $U$, and $\mathcal{E}_I$ is a fixed channel, independent of $U$ that only depends on the sample input states.

Interestingly, as we show in Appendix \ref{App:1}, Eq.(\ref{main_eq})   holds in a much more general setting: suppose we run the above algorithm with any other choice of  unitaries $W_a^\i(\lambda)$ and  $W^\o_a(\lambda)=(U\otimes I_a)  W^\i_a(\lambda) (U^\dag\otimes I_a)$ that couple the system to a qubit $a$, where $\lambda$ is chosen randomly from a finite set $\Lambda$, according to the probability distribution $p(\lambda)$. Then, Eq.(\ref{main_eq}) still holds for the channel 
\be
\mathcal{W}(\tau)=\sum_{\lambda\in\Lambda} p(\lambda)\ \Tr_a\big(W_a^\i(\lambda) [\tau\otimes |-\rangle\langle-|_a] {W_a^\i}^\dag(\lambda) \big)\ .
\ee  
As we explain in Appendix \ref{App:1}, this generalization is useful to extend  the algorithm to the case where the samples  contain mixed states.

\subsection{Runtime and error analysis}


To analyze the runtime and error in the algorithm, first, we ignore the error due to imperfect implementation of controlled reflections via density matrix exponentiation and study the circuit in Fig. (\ref{Fig}) with ideal controlled reflections.  Under this assumption,  Eq.(\ref{main_eq}) determines the fidelity of the output of the circuit with the ideal output.  As  shown in Appendix \ref{App:Runtime}, for any state $\rho$ with support restricted to $\mathcal{H}_\i$, the right-hand side of this bound can be rewritten as 
\begin{align}\label{bou9}
F(\mathcal{E}_U(\rho),U\rho U^\dag) &\ge \langle\phi_1^\i|\mathcal{W}^T(\rho)|\phi_1^\i\rangle= 1-\Tr\left(\mathcal{D}_\perp^T(\rho)\right)
\ , 
\end{align} 
where 
\be\label{Dperp}
\mathcal{D}_\perp(\cdot)=\frac{1}{K-1} \sum_{k=2}^K A_k (\cdot) A_k\ , 
\ee
and
\be
 A_k=P^\i_\perp R^\i(k) P^\i_\perp=P^\i_\perp-2 P^\i_\perp |\phi^\i_k\rangle\langle\phi^\i_k| P^\i_\perp \ ,
 \ee
where $P^\i_\perp$ is the projector to $(d-1)$-dimensional subspace of $\mathcal{H}_\i$ orthogonal to $|\phi^\i_1\rangle$.  For sufficiently large $T$,  the right-hand side of Eq.(\ref{bou9}) is approximately $1-\mathcal{O}(\lambda_\perp^T)$,
 where  $\lambda_\perp$ is the eigenvalue of  $\mathcal{D}_\perp$ with the largest magnitude.   Since $\mathcal{D}_\perp$ is a completely-positive self-dual map,  $\lambda_\perp$  is positive. Furthermore, proposition  \ref{prop2} together with Eq.(\ref{bou9}) imply that $\lim_{T\rightarrow \infty} \mathcal{D}_\perp^T(\rho)=0$ for all $\rho$, which in turn implies $0\le \lambda_\perp<1$.  Indeed, using the vectorization of $\mathcal{D}_\perp$, we find
 \be\label{lambda}
 \lambda_\perp=\frac{1}{K-1}\Big\|  \sum_{k=2}^K A_k\otimes A_k^\ast \Big\|_\infty< 1 \ ,
 \ee
where $A_k^\ast$ denotes the complex conjugate of $A_k$.

Using  Eq.(\ref{bou9}) in Appendix \ref{App:Runtime} we show that to achieve  
 the trace distance $\frac{1}{2}\|\mathcal{E}_U(\rho)-U\rho U^\dag\|_1 \le \epsilon_\text{id}$, it suffices to choose  
\beq\label{bound42}
T\ge  \frac{\ln (2\epsilon^{-2}_\text{id}[d-1]) }{1-\lambda_\perp} \ .
\eeq

 It turns out that in the actual algorithm for the universal quantum emulator, the dominant source of error arises from the imperfect implementation of the controlled reflections, which are realized using the provided copies of sample states.  
  In Appendix \ref{App:error2} we show that for initial state with support restricted to $\mathcal{H}_\i$, the action of unitary $U$ can be implemented with error  $\epsilon>0$ in the trace distance, using 
\beq\label{Ntot}
N_\text{tot}=\widetilde{\mathcal{O}}\big( \frac{d^2\times \epsilon^{-1}}{(1-\lambda_\perp)^2} \big)\ 
\eeq
total copies of sample states, and with 
\be\label{Ttot}
t_\text{tot}=\widetilde{\mathcal{O}}(N_\text{tot}\times \log D)
\ee
 elementary gates, where $\widetilde{\mathcal{O}}$ suppresses more slowly-growing terms (namely, multiplicative terms in $\log  \epsilon^{-1}$ and $\log d$). Note that the only place where $D$, the dimension of the Hilbert space, shows up in this analysis is in the implementation of  the controlled-reflections using the given copies of sample states. 
Also, it is worth noting that, as one should expect from the discussion in the previous section, the runtime and sample complexity depend on the spectral properties of the channel $\mathcal{W}$, which determine the rate of convergence to its fixed point $|\phi_1^\infty\rangle$.  





This analysis shows that  the gate and sample complexities of this algorithm should be close to optimal. In particular,  for general unitary $U$,  the state of each qubit in the desired output state $U|\psi\rangle$ depends non-trivially on the state of all other qubits in state $|\psi\rangle$. This means 
the minimum number of elementary gates needed to implement the algorithm is lower bounded by $\Omega(\log D)$. Furthermore, it can be easily seen that the lowest achievable runtime scales, at least, linearly with $d$, the dimension of the input subspace. (Indeed, just to check if the given state is inside the input subspace spanned by the sample states, one needs to interact with at least $d$ different sample states, which requires order $d$ gates.)  We conclude that the complexity of the proposed algorithm cannot be improved drastically.  
The following example shows the significance of the gap $1-\lambda_\perp$.


\subsection{Example:  Input subspace with dimension $d=2$}

For a system with the Hilbert space with arbitrary dimension $D$, consider an input subspace $\mathcal{H}_\i$ with dimension $d=2$ and two sample input states $|\phi_1^\i\rangle$ and $|\phi_2^\i\rangle$ in this subspace. Note that in this simple case the algorithm does not involve any
randomization: After controlled-reflection with respect
to $|\phi_1^\i\rangle$ we always apply the Hadamard on the ancillary
qubit, and then apply controlled-reflection with respect to state $|\phi_2^\i\rangle$. 

Then, using Eq.(\ref{lambda})  one can easily see that $\lambda_\perp=\cos^2 \theta $, 
where $|\langle\phi^\i_1|\phi^\i_2\rangle|=\cos({\theta}/{2})$, which means $\theta$ is the angle between  the Bloch vectors associated to $|\phi^\i_1\rangle$ and $|\phi^\i_2\rangle$.  Without loss of generality assume $|\phi_1^\i\rangle=|0\rangle$,
and define $|1\rangle$ to be the normalized vector in $\mathcal{H}_\i$ orthogonal to $|0\rangle$.  Then, the definition in Eq.(\ref{Dperp}) means that  
$$\mathcal{D}_\perp(\rho)= (\cos\theta)^2\  |1\rangle\langle 1|\rho |1\rangle\langle 1|\ .
$$ 
Together with Eq.(\ref{bou9}), this implies that for any input state $\rho$ restricted to this subspace
\be
F(\mathcal{E}_U(\rho),U\rho U^\dag) \ge 1-(\cos\theta)^{2T} \ .
\ee
For $\theta=0$ and $\theta=\pi$, we get $\lambda_\perp=(\cos\theta)^2=1$, which means 
the convergence does not happen. These two cases correspond to the identical and orthogonal sample states, respectively. Indeed, these are exactly the cases where the condition in Eq.(\ref{alg}) does not hold, and naturally the algorithm fails. On the other hand, for $\theta=\pi/2$ the perfect fidelity 1 can be  achieved with $T=1$. This corresponds to the case where the sample input states are coming from a pair of mutually unbiased bases.

To better appreciate the significance of the gap's value, $1-\lambda_\perp=(\sin\theta)^2$, let us suppose the goal is to emulate the unitary $\exp(i\alpha P)$ for an unknown $\alpha\in [0,2\pi)$, where $P=|\phi^\i_1\rangle\langle\phi^\i_1|$. Then, the given copies of sample output states $|\phi_k^\o\rangle: k=1,2$, should provide information about the unknown parameter $\alpha$ needed to implement this unitary. We can quantify this information using the Quantum Fisher Information \cite{Helstrom:book, Holevo:book,   braunstein1994statistical} of the output states with respect to the parameter $\alpha$.
 Under this family of unitaries, the input state $|\phi^\i_1\rangle$ remains invariant, up to a global phase, which means Quantum Fisher Information of the output state $|\phi^\o_1\rangle=e^{i\alpha}|\phi^\i_1\rangle$ is zero.  For the family of output state $|\phi^\o_2\rangle= \exp(i\alpha P)|\phi^\i_2\rangle$ Quantum Fisher Information  is determined by  the variance of the observable $P$ for state, as   
\be
I_\alpha=4\times \Big[\langle\phi^\i_2|P^2|\phi^\i_2\rangle-\langle\phi^\i_2|P|\phi^\i_2\rangle^2\Big]=1-\lambda_\perp\ .
\ee
Therefore, a smaller gap $1-\lambda_\perp$ implies that each copy of sample output states carries a lesser amount of Quantum Fisher Information about the unknown parameter $\alpha$. This, in turn, means that a larger number of copies of the sample states are needed to implement the unitary $\exp(i\alpha P)$ with a fixed accuracy. 

\subsection{Approximate transformations}

In the above discussions we always assumed that the unitary relation $U|\phi_k^\i\rangle=|\phi_k^\o\rangle$
 between the sample input and output states holds exactly. Now suppose this relation holds only approximately.  In particular, assume instead of these output states we are given output states $|\widetilde{\phi}_k^\o\rangle: l=1,\cdots, K$, satisfying 
\be\label{bound1}
|\langle\widetilde{\phi}^\o_k| U|\phi_k^\i\rangle|^2\ge 1-\delta^2  \ \ \ \ \ : k=1,\cdots, K\ ,
\ee 
for $\delta \ge 0$ and some unitary $U$. 
 Then,  if one uses these sample input-output pairs,  the circuit implements an approximate version of the unitary $U$. As before, let $\mathcal{E}_U$ be the channel  describing the output of the circuit in Fig. \ref{Fig}  when we run the circuit with the ideal controlled-reflections with respect to the sample states  
$\{|\phi^\i_k\rangle\}$ and  $\{|\phi_k^\o\rangle=U|\phi_k^\i\rangle\}$, and let $\widetilde{\mathcal{E}}$ be the modified  channel 
realized by the circuit that uses controlled-reflections with respect to states   
$\{|\widetilde{\phi}_k^\o\rangle\}$ instead of $\{|\phi_k^\o\rangle\}$. Then, as we show in Appendix \ref{App:approx}, for any input state $\rho$ the trace distance between the outputs of the two channels  is bounded by
\be\label{bound4}
\frac{1}{2}\left\|\mathcal{E}_U(\rho)-\widetilde{\mathcal{E}}(\rho)\right\|_1\le T\times 4\delta \ .
\ee
Therefore, assuming $T=\text{poly}(d)$, the algorithm implements unitary $U$ with the additional error $\delta\times \text{poly}(d)$.

\subsection{State compression algorithm}\label{Sec:comp}

In the above analysis, we did not need to explicitly determine how the state of ancilla qubits evolve during the algorithm.  Rather, we  argued that at the end of step (i)  (almost) all the information about the input state $|\psi\rangle$ is encoded in the  ancillary qubits. In the following, we determine the explicit form of this encoded state. This also  clarifies an interesting application  of this algorithm for state compression.

Let $|t\rangle_\textbf{a}=|1\rangle^{\otimes t}\otimes  |0\rangle^{\otimes (T-t)}$ be the state of qubits $a_1\cdots a_T$, in which $a_{t+1} \cdots a_{T}$ are all in state $|0\rangle$, and the rest of qubits are in state $|1\rangle$. Then, it can be easily seen that at the end of step (i) the joint state of the system and ancilla qubits $\textbf{a}=a_1\cdots a_T$ is given by
\beq\label{Eq_expan}
|\phi^\i_1\rangle\otimes \big(\sum_{t=0}^{T-1}  \langle\phi^\i_1|\psi(t,\textbf{k})\rangle \  |t\rangle_\textbf{a}\big) +|\psi(T,\textbf{k})\rangle\otimes |T\rangle_\textbf{a}  \ ,
\eeq
where the (unnormalized) vectors $|\psi(t,\textbf{k})\rangle$ are defined via the recursive relation 
\beq\label{rec2}
|\psi(t+1,\textbf{k})\rangle=R^\i(k_{t+1}) P^\perp |\psi(t,\textbf{k})\rangle \ ,
\eeq
and $|\psi(0,\textbf{k})\rangle=|\psi\rangle$. 
 The argument in the previous section implies that for initial state $|\psi\rangle\in\mathcal{H}_\i$, the typical  norm of $|\psi(t,\textbf{k})\rangle$ is exponentially small in $t$, and hence   for sufficiently large $T$ the last term in Eq.(\ref{Eq_expan}) is negligible.

  Note that this expansion only involves a $(T+1)$-dimensional subspace of the $2^T$-dimensional Hilbert space of the ancillary qubits.  That is, the information about state $|\psi\rangle$ is encoded in ancilla qubits using a unary coding.   Indeed, rather than $T$ qubits used in the circuit Fig.\ref{Fig}, the same algorithm can be algorithm can be realized using only $\lceil \log (T+1)\rceil$ ancillary qubits (See Appendix \ref{App:Sec:comp}). This modified version of the algorithm can be useful as a state compression algorithm, which is of independent interest:  Suppose two distant parties are given the same sample states in an unknown $d$-dimensional subspace of a Hilbert space of dimension $D\gg d$. Then, to transfer an unknown given state in this subspace, one party can use step (i) of this algorithm to compress the given state in $t^\ast=\text{poly}(\log(d))$ ancillary qubits, and send it to the other party. The receiver then can recover the original state, with error exponentially small in $t^\ast$

\begin{figure}[t]
  \includegraphics[scale=.35]{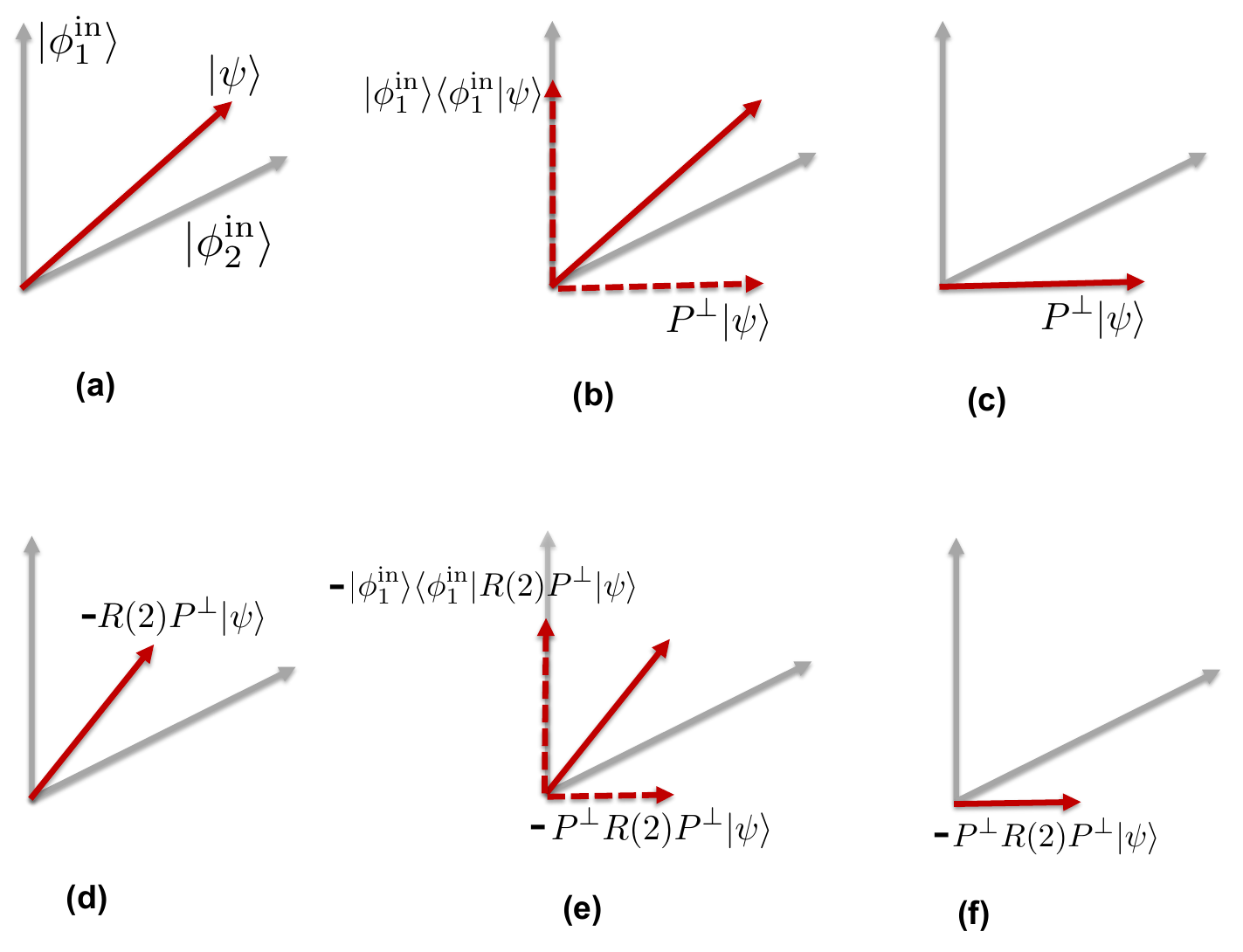
  }
  \caption{\textbf{A  geometric interpretation of the algorithm:}  As it can be seen from Eqs.(\ref{Eq_expan},\ref{rec2}), step (i) of the circuit in Fig.(\ref{Fig}) finds the coordinates of the input state $|\psi\rangle$ with respect to the sample input states $\{|\phi^\i_k\rangle\}$, and stores this information in the ancilla qubits.  
   In this Figure, we  
 present a geometric interpretation of step (i), for the case of a subspace with dimension $d=2$.    ($\textbf{a}$) Let $|\psi\rangle$ be an arbitrary state in a subspace with dimension $d=2$ spanned by $|\phi^\i_1\rangle$ and $|\phi^\i_2\rangle$. ($\textbf{b}$) Applying the first controlled-reflection with respect to $|\phi^\i_1\rangle$, the circuit finds $\langle\phi^\i_1|\psi\rangle$ and stores it as the coefficient of $|t=1\rangle_{\textbf{a}}$ (See Eq.(\ref{Eq_expan})). (\textbf{c}) Next, we need to determine the coordinates of the  residual vector $P^\perp|\psi\rangle$, the component of $|\psi\rangle$ orthogonal to $|\phi^\i_1\rangle$. ($\textbf{d}$) To achieve this, the algorithm first reflects the residual vector with respect to $|\phi^\i_2\rangle$, resulting in the vector $ R(2)P^\perp|\psi\rangle$.  Note that in the figure we have depicted $-R(2)P^\perp|\psi\rangle$.  
  ($\textbf{e}$) The assumption that $|\phi^\i_2\rangle$ is is not orthogonal  to $|\phi^\i_1\rangle$ now guarantees that this vector has again a non-zero component in the direction of $|\phi^\i_1\rangle$. Then,  again the algorithm finds their inner product and stores it in the state of ancillary qubits, as the coefficient of $|t=2\rangle_{\textbf{a}}$. ($\textbf{f}$) Next, the algorithm needs to determine the coordinates of the residual vector $ P^\perp R(2)P^\perp|\psi\rangle$. As we repeat this process the residual vector becomes exponentially smaller. Therefore, after few rounds of repetition, by ignoring the residual vector, we can still determine the original state $|\psi\rangle$ from the state of ancillary qubits, with a good approximation.}
  \label{Fig_Geometry}
\end{figure}


\subsection{Coordinates of the input state relative to the samples}



  The expansion found in Eq.(\ref{Eq_expan}) and Eq.(\ref{rec2}) indeed describes a general recursive method for specifying any vector $|\psi\rangle\in\mathcal{H}_\i$ in terms of the scalars 
  $$ \langle\phi^\i_1|\psi(t,\textbf{k})\rangle\ \ \ : t=0,\cdots , T-1\ ,$$
   which only depend on the relation between $|\psi\rangle$ and  states $|\phi^\i_{k_1}\rangle, \cdots , |\phi^\i_{k_T}\rangle$ (That is, they remain invariant under unitary transformations applied  on $|\psi\rangle$ and states $|\phi^\i_{k_1}\rangle, \cdots , |\phi^\i_{k_T}\rangle$). We can interpret these scalars as the \emph{coordinates} of vector $|\psi\rangle$ relative to the frame defined by states $\{|\phi^\i_{k}\rangle\}$.  Fig.\ref{Fig_Geometry} explains how this coordinate system works in the case of a 2-dimensional subspace.

Therefore, in this language, step (i) of the algorithm  is a circuit for finding the coordinates of the given state $|\psi\rangle$ relative to the input  samples.  Then, step (iv) reconstructs  the state with exactly the same coordinates relative to the frame defined by the output samples.  Note that because of the no-cloning theorem \cite{wootters1982single}, in order  to find the coordinates  of a quantum state and encode it in the ancillary qubits, we need to erase the state of system.  This method can be useful for other applications, where   instead of implementing operations on the system directly, we first find the coordinates of state of system with respect to other quantum states, implement an operation on the coordinates, and then transform the state back to the physical space.  
 
From this point of view, the given copies of each sample state can be interpreted as a Quantum Reference Frame (QRF) for \emph{directions} in the Hilbert space. A QRF usually refers to a reference frame for a physical degree of freedom, such as position or time, which is treated quantum mechanically \cite{QRF_BRS_07, MS11, Modes, gour2008resource, marvian2008building}. In contrast, here we are using the concept of QRF in a more abstract sense  
 (In particular, while for standard QRF's 
the relevant symmetry group is usually the groups such as $\text{SO(3)}$, $\text{U(1)}$ or $\mathbb{Z}_N$ \cite{QRF_BRS_07, MS11, Modes, gour2008resource, marvian2008building},  in this case the relevant symmetry group is   $\text{SU(D)}$).

\section{OTHER VARIATIONS OF THE ALGORITHM}\label{Sec:Disc}

As we explained at the end of Sec.\ref{Sec:alg}, the proposed algorithm  for the universal quantum emulator works based on a simple and general principle, namely emulating via coherent erasing. Hence, the algorithm can be generalized in several different ways based on this overarching principle.  Some of the possible generalizations are presented in Appendix \ref{App:Gen}, and are briefly discussed below.  \\


\noindent\textbf{Mixed Sample States:} We show that the algorithm can be generalized to  the case where the sample input-output sets contain mixed states. More precisely, as long as the sample input set, and hence the sample output set,  contains (at least) one state close to a pure state, we can still coherently erase the input state of the system and push it into this  pure state.   Then, we can emulate the action of the unknown unitary, using the same approach we used in the original algorithm. The main difference with the original version of the algorithm is that, instead of the controlled-reflections, in the case of mixed states we implement \emph{controlled-translations}  with respect to the sample states, i.e.,  the unitaries $|0\rangle\langle 0|\otimes I+|1\rangle\langle 1|\otimes e^{-i t \sigma_k}$, for a random value of $t$ and a sample state $\sigma_k$.\\


 \noindent\textbf{Prior knowledge of samples:}    In Appendix \ref{App:Gen}, we also present a more efficient algorithm which works under certain extra assumptions about the  sample input states. Namely, this version of the algorithm assumes the input samples are states in an (unknown)  orthonormal basis for the input subspace, plus one or more states in the Fourier conjugate basis. The working principle behind this algorithm is again emulating via coherent erasing. In this algorithm the state of system is erased via a coherent measurement in the orthonormal basis followed by another coherent measurement in the conjugate basis. It is also worth noting that if some input or output sample states are known, then, rather than  using sample states and the density matrix exponentiation method, one can directly implement the corresponding controlled-reflections.\\



 \noindent\textbf{An algorithm for projective measurements:}  Suppose in the algorithm presented in Fig.(\ref{Fig})  we use the same set of states as both the input and the output samples. In this case, the algorithm basically performs the two-outcome projective measurement that projects the given state into the subspace spanned by the sample states, or its complement.  Any arbitrary projective measurement  can be implemented as a sequence of these two-outcome projective measurements. 
 
This approach, however, only works if the set of sample states in each subspace generates the full matrix algebra in that subspace, as stated in Eq.(\ref{comm}).  
   But in this context, requiring this assumption is unjustifiable: To specify a projective measurement one only needs to specify the subspaces   corresponding to different outcomes, and this can be achieved by sets of states that do not necessarily satisfy this constraint.

   To address this issue, in Appendix \ref{App:meas} we present a different efficient algorithm for performing projective measurements, which does not require this extra assumption. This algorithm also uses random controlled-reflections about the given sample states.

\section{Future Directions}
In the following, we outline potential future directions and applications of Universal Quantum Emulator (UQE).   \\

\begin{figure}[t]
  \includegraphics[scale=.4]{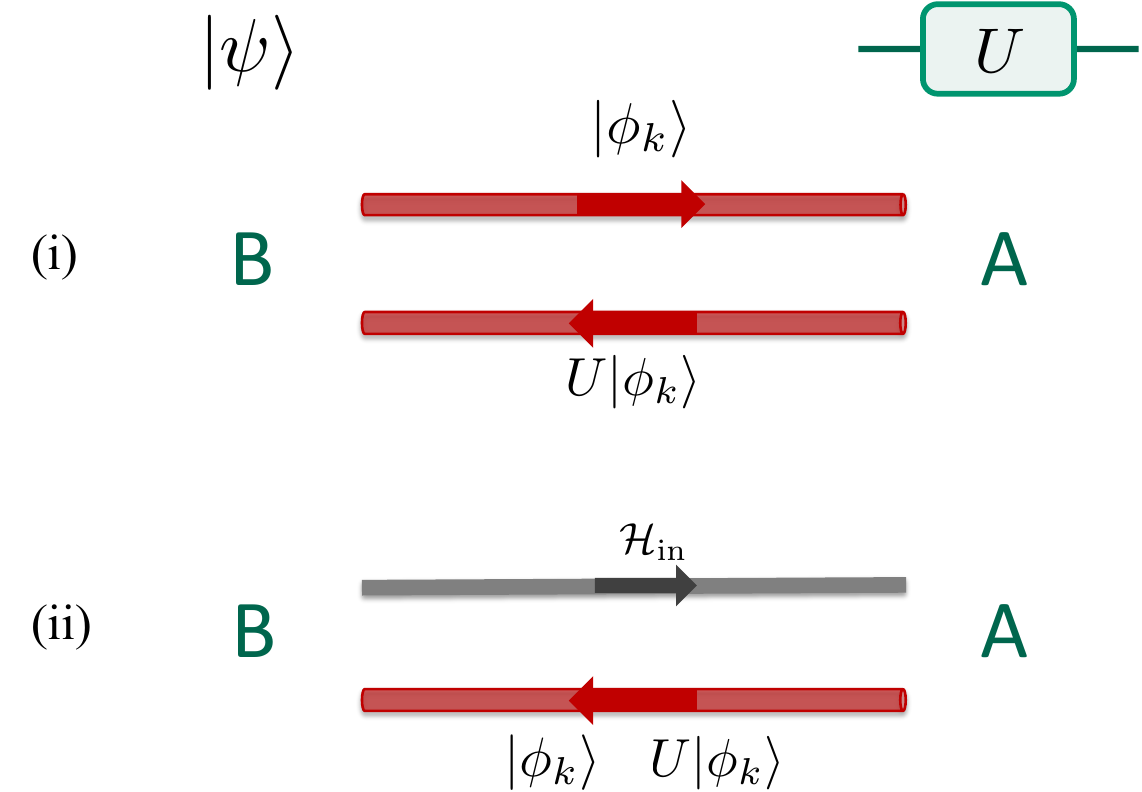}
  \caption{\textbf{Secure Quantum Computation:}  Bob has state $|\psi\rangle$, Alice has the unitary $U$ and the
  goal is to transform Bob's system to state $U|\psi\rangle$.  Alice and Bob want to reveal minimum possible information about their state and unitary to the other party. Here, we describe two slightly different protocols based on Universal Quantum Emulator: $\textbf{(i)}$ Bob sends random states $\{|\phi_k\rangle\}$ to Alice in a random order.  Alice applies unitary $U$ on them and send them back to Bob. Bob uses these sample output states to run the UQE algorithm on state $|\psi\rangle$. $\textbf{(ii)}$  Bob chooses a random subspace $\mathcal{H}_\i$ which contains state $|\psi\rangle$ and sends a classical description of this subspace to Alice.  Alice chooses a larger random subspace  $\widetilde{\mathcal{H}}_\i$ such that ${\mathcal{H}}_\i  \subset \widetilde{\mathcal{H}}_\i$. Then, she chooses random states $|\phi_k\rangle\in{\mathcal{H}}_\i $ apply unitary $U$ on them and send copies of states $\{|\phi_k\rangle\} $ and 
  $\{U|\phi_k\rangle\} $ to Bob. 
   Bob uses these samples to run the UQE algorithm on state $|\psi\rangle$.}
  \label{FigSecure}
\end{figure}

\subsection{Secure Quantum Computation}
A potential area of application of UQE is in the context of secure quantum computation.  
Suppose a client (Bob) has state $|\psi\rangle$ and a server (Alice) has the unitary $U$ and the goal is to transform Bob’s system to state $U|\psi\rangle$.  However, Alice does not want to reveal any information about unitary $U$ to Bob and Bob   
does not want to reveal any information about state  $|\psi\rangle$  to Alice. In the following, we briefly describe two  slightly different protocols which employ the UQE algorithm to achieve these goals (See also Fig.\ref{FigSecure}). \\
 
\noindent\textbf{Protocol I:} (1) Bob picks a random subspace $\mathcal{H}_\i$ that contains state $|\psi\rangle$. (2) He chooses $K$ random states $\{|\phi^\i_k\rangle\ : k=1,\cdots, K\}$, with $K$ larger than or equal to the dimension of subspace $\mathcal{H}_\i$. Then, he generates $n$ copies of each state $|\phi^\i_k\rangle$ and sends all $n\times K$ states to Alice in a random order. (3) Alice runs the unitary $U$ on all these states and send them back to Bob. (4) Bob uses these
states as the sample output states to run the UQE algorithm on the input state $|\psi\rangle$. Note that since he knows
what are the sample input states $\{|\phi^\i_k\rangle\ : k=1,\cdots, K\}$ 
in the first step of the algorithm he can implement the controlled-reflections with respect to these states directly without using density matrix exponentiation.

In this protocol, there is a trade-off between
the amount of information that leaks to the server Alice about state $|\psi\rangle$, on one hand, and the communication and complexity cost of running the scheme, on the other hand: If Bob chooses a
low-dimensional subspace $\mathcal{H}_\i$, then the server obtains a
considerable amount of information about the unknown
state $|\psi\rangle$. In other words, she can create states with a large overlap with $|\psi\rangle$ (for example, by mixing the given samples). On the other hand, choosing a larger subspace requires the client and server to exchange a greater number of systems  and run a longer circuit. \\

\noindent\textbf{Protocol II}: (1) Bob picks a random subspace $\mathcal{H}_\i$ which
contains state $|\psi\rangle$, and sends the classical description of
this subspace to Alice. (2) Alice randomly chooses a
subspace $\widetilde{\mathcal{H}}_\i$ which includes $\mathcal{H}_\i$, i.e., $\mathcal{H}_\i\subset \widetilde{\mathcal{H}}_\i$  but
is sufficiently 
larger than $\mathcal{H}_\i$. She chooses $K$ random
sample input states $\{|\phi^\i_k\rangle\ : k=1,\cdots, K\}$,  where $K$ is 
larger than or equal to the dimension of $\widetilde{\mathcal{H}}_\i$. Then, she generates $2n$ copies of each state, apply $U$ on $n$ copies and sends all $2n\times K$ states to Bob. (3) Bob uses these states to run the UQE algorithm on the input state $|\psi\rangle$. 

Again, note that there is a trade-off  between the leaked information about state $|\psi\rangle$ and the unitary $U$, on one hand, and the required amount of quantum communication and the complexity of the protocol, on the other hand. Understanding the efficiency and security of these protocols is an interesting open question and  requires further investigation.\\

\subsection{Quantum Resource Theories}
Another natural area of application for UQE and its variants discussed in  Appendix \ref{App:Gen}, is in the context of quantum resource theories \cite{chitambar2019quantum}. Any quantum resource theory  is defined by a restriction on the set of operations that can be realized with negligible cost, called the \emph{free} operations of the resource theory. Consequently, there is often an interest in determining the cost of implementing a general non-free unitary transformation in terms of (i) quantum states containing the resource under consideration, and (ii) other possible available non-free operations. 

For example, in the resource theory of thermodynamics  \cite{FundLimitsNature, brandao2013resource,  lostaglio2015quantumPRX}  it is usually assumed that the set of free unitaries are energy-conserving unitaries, i.e., those that conserve the total intrinsic Hamiltonians of the systems. Similarly, in the resource theory of asymmetry \cite{gour2008resource, MS11} and the related area of quantum reference frames \cite{QRF_BRS_07, marvian2008building},  the free unitaries are those that are invariant under the representation of the symmetry under consideration.

The universal quantum emulator can serve as a universal protocol in quantum resource theories for implementing non-free unitary transformations.  This protocol requires consuming resource quantum states, and realizing  unitaries that are either free or non-free but accessible,  depending on the resource theory.

For such applications, in the following, we provide a comprehensive list of all the operations and states necessary for implementing this algorithm.  
It is also worth noting that while the algorithm in Fig.\ref{Fig} has the advantage that it can be implemented without knowledge of the sample input/output states, if one does have such information, then it can be implemented more efficiently using the circuits discussed in Appendix \ref{App:Gen}.

\begin{enumerate}

\item \textbf{Required Operations:} Implementing the protocol requires  the following unitary transformations:
\begin{enumerate} 
\item Controlled-Controlled-SWAP gates, where both control qubits are ancillary qubits.   
\item Single-qubit gates on the ancillary qubits.
 \item One single-qubit measurement on an ancilla qubit, to determine if the given input state is in the input subspace $\mathcal{H}_\i$ or not.  
 \end{enumerate}
 \item \textbf{Consumed Resource States:}  The protocol fully consumes 
\begin{enumerate}
\item One copy of one sample output state $|\phi_1^\o\rangle$.
\end{enumerate}
\item\textbf{Approximate Catalysts:} Unlike the above sample output state, which is fully consumed by the protocol, most sample states and ancilla qubits needed to implement the protocol will be returned to their initial state with a small error, which vanishes in the asymptotic regime. Therefore, to highlight this distinction, we refer to them as \emph{approximate catalysts}. However, it is important to note that they are not true catalysts because their final state is not the same as their initial state.
\vspace{2mm}
\begin{enumerate} 
 \item Ancilla qubits which can be all initialized in a fixed state, say $|0\rangle$.
 \item Sample input-output states.
\end{enumerate}

\item  \textbf{Returned Resource State:} In addition to implementing the desired unitary, the protocol also generates 

\begin{enumerate}
\item One copy of the sample input state $|\phi_1^\i\rangle$.
\end{enumerate}

\end{enumerate}

Interestingly, all the required unitaries are free in the resource theories of quantum thermodynamics and asymmetry under the assumption that the ancilla qubits have a trivial Hamiltonian and a trivial representation of the symmetry, respectively. Furthermore, under this assumption,  in the resource theory of asymmetry, ancilla qubits in arbitrary state are also free.

\subsection{Conclusion}

In this paper, we presented an efficient algorithm for emulating unknown unitary transformations and projective measurements. The algorithm employs a novel randomized technique based on a simple principle, which we termed \emph{emulating via coherent erasing}.   We discussed some potential applications of this algorithm in the contexts of quantum resource theories and secure computation. We also showed that even when the relationship between sample input-output pairs can only be approximated by a unitary transformation, the algorithm implements a transformation that is close to this unitary.  This capability could be useful for applications in quantum machine learning \cite{biamonte2017quantum}  and, for instance, in addressing the unitary Procrustes problem \cite{lloyd2020quantum}—finding a unitary transformation that approximates the conversion of input data to a set of output data.

\section{Acknowledgments}
We would like to thank Alireza Seif for helpful discussion and for bringing Ref. \cite{pichler2016measurement} to our attention.  This work was supported by grants AFOSR No 6929347 and ARO MURI No 6924605 . 
 \bibliography{Ref_v11, Ref_2021_v3}

\newpage

\onecolumngrid

\newpage
\appendix

\maketitle
\vspace{-5in}
\begin{center}

\Large{\textbf{Supplementary Material}} \end{center}
\section{Fidelity of emulation for the generalized algorithm}  \label{App:1}

In this section, we present a generalized version of the algorithm discussed in the paper. We also prove Eq.(\ref{main_eq}) for this generalized algorithm.

\subsection{Preliminaries}
Let $W_a^\i(\lambda)$ be an arbitrary unitary acting on the system and ancilla qubit $a$, where the parameter $\lambda$ is an element of a finite set $\Lambda$. Let   
\beq
W_a^\o(\lambda)=(U\otimes I_a) W_a^\i(\lambda) (U\otimes I_a)\ .
\eeq
Consider states $|\phi^\i\rangle$, and   $|\phi^\o\rangle=U|\phi^\i\rangle$.  Let $R(\phi^\i)\equiv e^{i\pi |\phi^\i\rangle\langle \phi^\i|}$ be reflection about the state $|\phi^\i\rangle$, and 
\beq
R_c(\phi^\i)=|0\rangle\langle 0|_c\otimes I+|1\rangle\langle 1|_c\otimes e^{i\pi |\phi^\i\rangle\langle \phi^\i|}
\eeq
be the controlled-reflection  about state $|\phi^\i\rangle$, acting on the system and ancilla qubit $c$.

\subsection{The generalized algorithm}
Here we list the  steps of the algorithm. The quantum circuit for this algorithm is presented in Fig.(\ref{Fig2}). \\ 

\begin{itemize}

\item [\textbf{(i)}] Let $\lambda_1,\cdots , \lambda_T$ be $T$   random elements of a finite set $\Lambda$, chosen independently according to the distribution $p(\lambda)$. 
Consider $T$  ancillary qubits labeled by $a_1,\cdots , a_T$, all initialized in the state $|-\rangle$. Apply the unitary $W_{a_1}^\i(\lambda_1)$ on the system and the ancillary qubit  $a_1$,  unitary $W_{a_2}^\i(\lambda_2)$ on the system and $a_2$, and so on, until the last unitary $W_{a_T}^\i(\lambda_T)$, which is applied to the system and  $a_T$.  


\item [\textbf{(ii)}] Apply the controlled-reflection $R_c(\phi^\i)$ on the system and qubit $c$,  initially prepared in state $|-\rangle$. Then, apply a Hadamard gate to the qubit $c$ and measure it in the computational basis.


\item [\textbf{(iii)}] Swap a copy of state $|\phi^\o\rangle$ with the state of system, which prepares the system in state $|\phi^\o\rangle$. 


\item [\textbf{(iv)}] 
Recall  $\lambda_1,\cdots , \lambda_T$ chosen in step (i) and apply the sequence of unitaries ${W_{a_T}^\o}^\dag({\lambda_T}),\cdots , {W_{a_1}^\o}^\dag({\lambda_1})$ on the system and the $T$ ancillary qubits     $a_T \cdots a_1 $, in the following order: First apply $ {W_{a_T}^\o}^\dag({\lambda_T})$ to the system and  qubit $a_T$, then apply   $ {W_{a_{T-1}}^\o}^\dag({\lambda_{T-1}})$ to the system and $a_{T-1}$, and so on, until the last unitary  $ {W_{a_1}^\o}^\dag({\lambda_1})$ which is applied on the system and qubit $a_{1}$.

\end{itemize}

\subsection{Fidelity of emulation (Proof of Eq.(\ref{main_eq}))}
Recall the definition of (Uhlmann) Fidelity, between two density operator $\sigma_1$  and $\sigma_2$,
\beq
F(\sigma_1,\sigma_2)=\|\sqrt{\sigma_1} \sqrt{\sigma_2}\|_1=\Tr\left(\sqrt{\sqrt{\sigma_1}\sigma_2 \sqrt{\sigma_1}}  \right)\ ,
\eeq
where $\|\cdot\|_1$ is $l_1$ norm, defined as the sum of the singular values of the operator. In the following,  we prove that 
\begin{theorem}\label{Thm_fid}
  Let $\mathcal{E}_U$ be the quantum channel that describes the overall effect of of the algorithm presented in Fig.\ref{Fig2}  on the state of system, in the case where  we do not post-select based on the outcome of measurement in step (ii) (In other words, we do not care if erasing has been successful or not). Then, for any input state  $\rho$ the Uhlmann fidelity of $\mathcal{E}_U(\rho)$ and the desired state $U\rho U^\dag$ satisfies
\beq
\text{F}(\mathcal{E}_U(\rho), U\rho U^\dag)\ge p_\text{erase}(\rho) = \langle\phi^\i|\mathcal{W}^T(\rho)|\phi^\i\rangle\ ,
\eeq
where
\beq
\mathcal{W}(\tau)=\sum_{\lambda\in\Lambda} p(\lambda) \  \Tr_a(W^\i(\lambda) [\tau\otimes |-\rangle\langle-|_a] {W^\i}^\dag(\lambda) )\ ,
\eeq
and  $p_\text{erase}(\rho)$ is the probability that at the end of step (i) the system is found in state $|\phi^\i\rangle$, or, equivalently the probability that the measurement in step(ii) returns outcome b=0.  Furthermore, if we post-select to the cases where $b=0$, which means the erasing has been successful, then for pure input state $\rho$ the fidelity of the output state $\tilde{\rho}_\text{post}$ with the desired state $U\rho U^\dag$ is lower bounded by 
 \beq
\text{F}(\tilde{\rho}_\text{post}, U\rho U^\dag)\ge \sqrt{p_\text{erase}(\rho) } = \sqrt{\langle\phi^\i|\mathcal{W}^T(\rho)|\phi^\i\rangle}\ .
\eeq
\end{theorem}

\begin{figure*}
  \includegraphics[width=\textwidth,height=7cm]{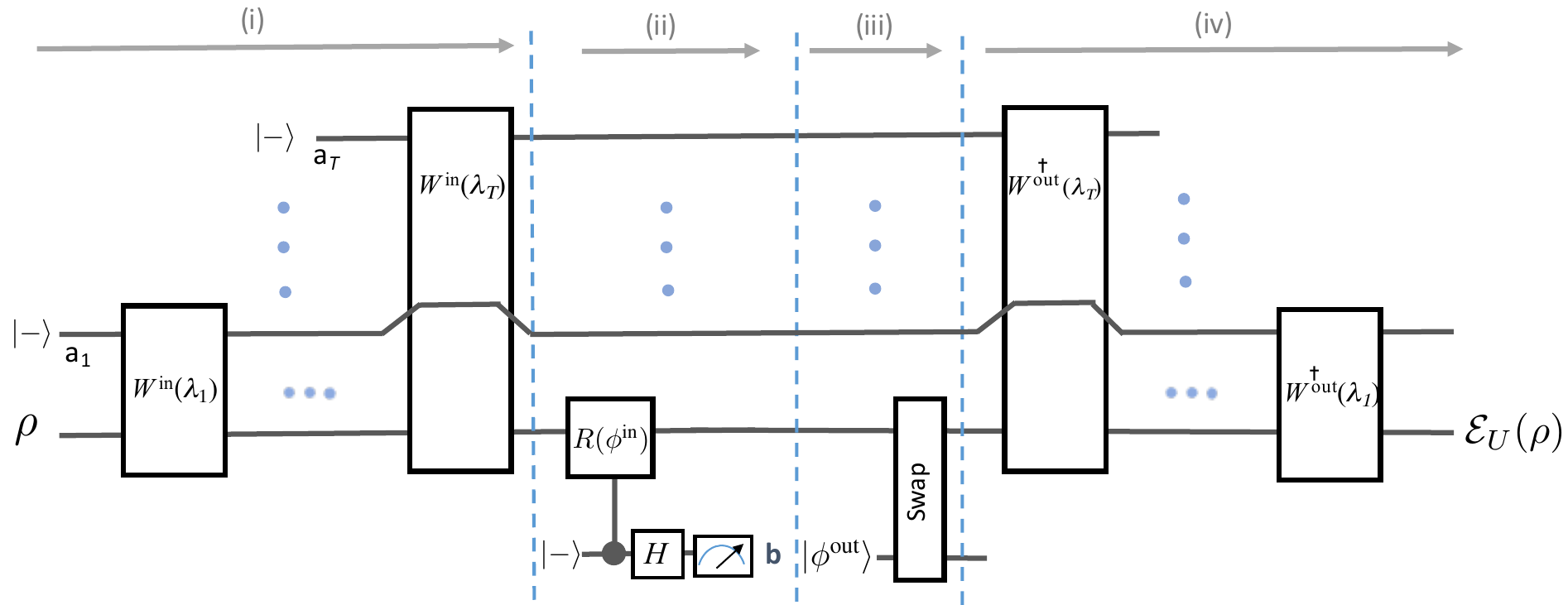}
\caption{The generalized quantum circuit for emulating  unitary $U$. .   
   \label{Fig2}}
\end{figure*}


\begin{proof}
We prove the theorem for the case of initial pure state $\rho=|\psi\rangle\langle\psi|$. The result for general mixed states follows from the joint concavity of Uhlmann fidelity.

Let $|\Theta(\boldsymbol{\lambda})\rangle$ be the joint state of the system and the ancillary qubits $a_1\cdots a_T$ at the end of step (i), for a particular choice of random parameters $\boldsymbol{\lambda}=(\lambda_1,\cdots , \lambda_T)$. This state is given by
\beq\label{def56}
|\Theta(\boldsymbol{\lambda})\rangle=\Big[W_{a_T}^\i(\lambda_T)\cdots W_{a_1}^\i(\lambda_1)\Big] \otimes (|\psi\rangle\otimes |-\rangle^{\otimes T})\ . 
\eeq
Then, if we ignore the outcome of the measurement in step (ii), at the end of step (iii) the joint state of the system and ancillary qubits $a_1\cdots a_T$, for this particular choice of random parameters $\boldsymbol{\lambda}$ is given by
\beq
|\phi^\o\rangle\langle\phi^\o|\otimes \Tr_S(|\Theta(\boldsymbol{\lambda})\rangle\langle\Theta(\boldsymbol{\lambda})|)\ ,
\eeq
where the partial trace is over the  Hilbert space of the main system.  Therefore, at the end of the algorithm the final state is given by 
\beq\label{34t}
\left[{W_{a_1}^\o}^\dag(\lambda_1)\cdots {W_{a_T}^\o}^\dag(\lambda_T)\right] \Big[|\phi^\o\rangle\langle\phi^\o|\otimes \Tr_S(|\Theta(\boldsymbol{\lambda})\rangle\langle\Theta(\boldsymbol{\lambda})|)\Big] \left[{W_{a_T}^\o}(\lambda_T) \cdots {W_{a_1}^\o}(\lambda_1) \right]\ ,
\eeq
and the average joint state of the system and qubits $a_1\cdots a_T$ is given by  
\beq\label{glob_s}
\sigma_\text{fin}=\sum_{\lambda_1,\cdots, \lambda_T\in\Lambda} p(\lambda_1)\cdots p(\lambda_T)\  \left[{W_{a_1}^\o}^\dag(\lambda_1)\cdots {W_{a_T}^\o}^\dag(\lambda_T)\right] \Big[|\phi^\o\rangle\langle\phi^\o|\otimes \Tr_S(|\Theta(\boldsymbol{\lambda})\rangle\langle\Theta(\boldsymbol{\lambda})|)\Big] \left[{W_{a_T}^\o}(\lambda_T) \cdots {W_{a_1}^\o}(\lambda_1) \right]\ ,
\eeq
where we have averaged over  iid random parameters $(\lambda_1,\cdots , \lambda_T)$, where each $\lambda_i\in\Lambda$ happens with probability $p(\lambda_i)$. 

Next, we consider the fidelity  between the actual final state of system and ancilla $\sigma_\text{fin}$ and their ideal final state 
 $U|\psi\rangle|-\rangle^{\otimes T}$. In particular, the squared of this global fidelity is given by 
 \be
  F_\text{glob}^2= \langle-|^{\otimes T} \langle\psi| U^\dag\sigma_\text{fin}U|\psi\rangle|-\rangle^{\otimes T}\ .
 \ee
To calculate this fidelity note that by applying the unitary $U$ on the Hilbert space of system 
 on both sides of Eq.(\ref{def56}), we obtain 
 \beq
(U\otimes I_\textbf{a})|\Theta(\boldsymbol{\lambda})\rangle=\Big[W_{a_T}^\o(\lambda_T)\cdots W_{a_1}^\o(\lambda_1)\Big] (U |\psi\rangle)\otimes |-\rangle^{\otimes T}\ , 
\eeq
  where we have used  the fact that $(U\otimes I_a)  W^\i_a(\lambda) (U^\dag\otimes I_a)=W^\o_a(\lambda)$. Then, for each choice of $\boldsymbol{\lambda}$, the squared fidelity of this ideal state $(U |\psi\rangle)\otimes |-\rangle^{\otimes T}$ with the actual output state in Eq.(\ref{34t}) is given by
 \be
\langle\Theta(\boldsymbol{\lambda})| \Big[|\phi^\i\rangle\langle\phi^\i|\otimes \Tr_S(|\Theta(\boldsymbol{\lambda})\rangle\langle\Theta(\boldsymbol{\lambda})|)\Big] |\Theta(\boldsymbol{\lambda})\rangle\ ,
 \ee 
 where we have used the fact  that $U|\phi^\i\rangle=|\phi^\o\rangle$. Similarly, the  squared fidelity of the ideal state $(U |\psi\rangle)\otimes |-\rangle^{\otimes T}$ with the averaged state $\sigma_\text{fin}$ is given by
 \bes\label{ryr}
\begin{align}
F_\text{glob}^2&= \langle-|^{\otimes T} \langle\psi| U^\dag\sigma_\text{fin}U|\psi\rangle|-\rangle^{\otimes T}\\  &=\sum_{\lambda_1,\cdots, \lambda_T} \ p(\lambda_1)\cdots p(\lambda_T) \ \langle\Theta(\boldsymbol{\lambda})| \Big[|\phi^\i\rangle\langle\phi^\i|\otimes \Tr_S(|\Theta(\boldsymbol{\lambda})\rangle\langle\Theta(\boldsymbol{\lambda})|)\Big] |\Theta(\boldsymbol{\lambda})\rangle\ .
\end{align}
\ees
Next, to establish a lower bound on the squared fidelity we note that 
\begin{align}
\Tr_S(|\Theta(\boldsymbol{\lambda})\rangle\langle\Theta(\boldsymbol{\lambda})|) \ge \langle\phi^\i|\Theta(\boldsymbol{\lambda})\rangle\langle\Theta(\boldsymbol{\lambda}) |\phi^\i\rangle_S \ ,
\end{align}
which in turn implies
\be
|\phi^\i\rangle\langle\phi^\i|\otimes \Tr_S(|\Theta(\boldsymbol{\lambda})\rangle\langle\Theta(\boldsymbol{\lambda})|) \ge  \big[|\phi^\i\rangle\langle\phi^\i|\otimes I_\textbf{a}\big] |\Theta(\boldsymbol{\lambda})\rangle\langle\Theta(\boldsymbol{\lambda})|\big[ |\phi^\i\rangle\langle\phi^\i|\otimes I_\textbf{a} \big]\ ,
\ee
where  $I_S$ and $ I_\textbf{a} $ are, respectively, the identity operators on the system and the ancillary qubits $a_1\cdots a_T$.  Putting this into Eq.(\ref{ryr}) we arrive at
\bes\label{eqw12}
\begin{align}
F_\text{glob}^2&=\sum_{\lambda_1,\cdots, \lambda_T} p(\lambda_1)\cdots p(\lambda_T)\  \langle\Theta(\boldsymbol{\lambda})| \Big[|\phi^\i\rangle\langle\phi^\i|\otimes \Tr_S(|\Theta(\boldsymbol{\lambda})\rangle\langle\Theta(\boldsymbol{\lambda})|)\Big] |\Theta(\boldsymbol{\lambda})\rangle\\ &\ge  \sum_{\lambda_1,\cdots, \lambda_T} p(\lambda_1)\cdots p(\lambda_T)\  \langle\Theta(\boldsymbol{\lambda})| \big[|\phi^\i\rangle\langle\phi^\i|\otimes I_\textbf{a}\big] |\Theta(\boldsymbol{\lambda})\rangle\langle\Theta(\boldsymbol{\lambda})|\big[ |\phi^\i\rangle\langle\phi^\i|\otimes I_\textbf{a} \big] |\Theta(\boldsymbol{\lambda})\rangle\ \\ &=  \sum_{\lambda_1,\cdots, \lambda_T} p(\lambda_1)\cdots p(\lambda_T)\  \Big(\langle\Theta(\boldsymbol{\lambda})| \big[|\phi^\i\rangle\langle\phi^\i|\otimes I_\textbf{a}\big] |\Theta(\boldsymbol{\lambda})\rangle\Big)^2\ \label{wrf}\ .
\end{align}
\ees
Then, using the fact that the  variance of any random variable is non-negative, we find that
\bes\label{eqw13}
\begin{align}
F_\text{glob}^2&\ge  \sum_{\lambda_1,\cdots, \lambda_T} p(\lambda_1)\cdots p(\lambda_T)\  \Big(\langle\Theta(\boldsymbol{\lambda})| \big[|\phi^\i\rangle\langle\phi^\i|\otimes I_\textbf{a}\big] |\Theta(\boldsymbol{\lambda})\rangle\Big)^2\ \\ &\ge \Big(\sum_{\lambda_1,\cdots, \lambda_T} p(\lambda_1)\cdots p(\lambda_T)\  \langle\Theta(\boldsymbol{\lambda})| \big[|\phi^\i\rangle\langle\phi^\i|\otimes I_\textbf{a}\big] |\Theta(\boldsymbol{\lambda})\rangle\Big)^2
\\ 
&=p^2_\text{erase}(|\psi\rangle\langle\psi|)\ , 
\end{align}
\ees
where 
\be
p_\text{erase}(|\psi\rangle\langle\psi|)=\langle\phi^\i| \big[\sum_{\lambda_1,\cdots, \lambda_T} p(\lambda_1)\cdots p(\lambda_T)\ \Tr_\textbf{a}(|\Theta(\boldsymbol{\lambda})\rangle\langle\Theta(\boldsymbol{\lambda})|)\big]|\phi^\i\rangle\ ,
\ee
 is the probability that at the end of step (i) the reduced state of system is in state $|\phi^\i\rangle$. 

Finally, we note that the final state of system is obtained from state $\sigma_\text{fin}$, by  tracing over the ancillary qubits, i.e., $\mathcal{E}_U(|\psi\rangle\langle\psi|)=\Tr_\text{\bf{a}}(\sigma_\text{glob})$\ .
Then, using the  monotonicity of  Uhlmann fidelity under the partial trace we find
\beq
\text{F}(\mathcal{E}_U(|\psi\rangle\langle\psi|), U|\psi\rangle\langle\psi| U^\dag)\ge F_\text{glob} \ge p_\text{erase}(|\psi\rangle\langle\psi|)\ .
\eeq
Using the joint concavity of Uhlmann fidelity one can easily extend this bound to arbitrary initial state $\rho$. 

Next, using the definition $|\Theta(\boldsymbol{\lambda})\rangle=\Big[W_{a_T}^\i(\lambda_T)\cdots W_{a_1}^\i(\lambda_1)\Big] |\psi\rangle\otimes |-\rangle^{\otimes T}$, it can be easily shown that 
\bes
\begin{align}
\sum_{\lambda_1,\cdots, \lambda_T} p(\lambda_1)\cdots p(\lambda_T)\Tr_\textbf{a}(|\Theta(\boldsymbol{\lambda})\rangle\langle\Theta(\boldsymbol{\lambda})|)= \mathcal{W}^T(|\psi\rangle\langle\psi|) ,
\end{align}
\ees
where 
\beq
\mathcal{W}(\tau)=\sum_{\lambda\in\Lambda} p(\lambda)\  \Tr_a\Big( W_a^\i(\lambda)\Big[ \tau \otimes |-\rangle\langle-|_a \Big] {W_a^\i(\lambda)}^\dag\Big) \ ,
\eeq
and the partial trace is over qubit $a$. Therefore, we conclude that for any initial state $\rho$ it holds that $\text{F}(\mathcal{E}_U(\rho), U\rho U^\dag)\ge  p_\text{erase}(\rho)= \langle\phi^\i|  \mathcal{W}^T(\rho)|\phi^\i\rangle .$

In the above analysis, we ignored the outcome of measurement in step (ii), which means that we do not care if the erasing has been successful or not. Now suppose we postselect to the cases where the   erasing has been successful.  In this case after the measurement 
 the joint  state of the system and ancillary qubits $a_1\cdots a_T$ is given by
 \begin{align}
\frac{1}{p_\text{erase}(|\psi\rangle\langle\psi|)}\times  \sum_{\lambda_1,\cdots, \lambda_T} p(\lambda_1)\cdots p(\lambda_T) \big[|\phi^\i\rangle\langle\phi^\i|\otimes I_\textbf{a}\big] |\Theta(\boldsymbol{\lambda})\rangle\langle\Theta(\boldsymbol{\lambda})|\big[ |\phi^\i\rangle\langle\phi^\i|\otimes I_\textbf{a} \big]  ,
\end{align}
where the factor $1/p_\text{erase}(|\psi\rangle\langle\psi|)$ is due to the postselection.  Then, using an argument similar to the one we used before, we can show that  the the squared of the fidelity of the final joint state of the system and qubits $a_1\cdots a_T$ at the end of the algorithm, with state $U|\psi\rangle\otimes |-\rangle^{\otimes T}$   is given by  
\bes
\begin{align}
{F^{\text{post}}_\text{glob}}^2&=\frac{1}{p_\text{erase}(|\psi\rangle\langle\psi|)}\times \sum_{\lambda_1,\cdots, \lambda_T} p(\lambda_1)\cdots p(\lambda_T)\ \langle\Theta(\boldsymbol{\lambda})| \big[|\phi^\i\rangle\langle\phi^\i|\otimes I_\textbf{a}\big] |\Theta(\boldsymbol{\lambda})\rangle\langle\Theta(\boldsymbol{\lambda})|\big[ |\phi^\i\rangle\langle\phi^\i|\otimes I_\textbf{a} \big] |\Theta(\boldsymbol{\lambda})\rangle\ 
\\ &=\frac{1}{p_\text{erase}(|\psi\rangle\langle\psi|)}\times \sum_{\lambda_1,\cdots, \lambda_T} p(\lambda_1)\cdots p(\lambda_T)\ \Big(\langle\Theta(\boldsymbol{\lambda})| \big[|\phi^\i\rangle\langle\phi^\i|\otimes I_\textbf{a}\big] |\Theta(\boldsymbol{\lambda})\rangle\Big)^2\ ,
\end{align}
\ees
Then, using  Eq.(\ref{eqw13}),  we find
\begin{align}
{F^{\text{post}}_\text{glob}}^2&\ge \frac{p^2_\text{erase}(|\psi\rangle\langle\psi|)}{p_\text{erase}(|\psi\rangle\langle\psi|)} =p_\text{erase}(|\psi\rangle\langle\psi|) .
\end{align}
This means that if we postselect to the cases where erasing has been successful, then 
\beq
F(U|\psi\rangle\langle \psi| U^\dag, \tilde{\rho}_\text{post} )\ge F^{\text{post}}_\text{glob} \ge \sqrt{p_\text{erase}(|\psi\rangle\langle\psi|)} .
\eeq

\end{proof}
\newpage

\section{Error analysis}\label{App:Runtime}

In this section,  we first study the error in the output of  the algorithm, assuming  we run the quantum circuit in Fig. \ref{Fig} with ideal controlled-reflections $R_a(k)$. Then, we consider the error in the actual algorithm,  where we use the given copies of sample states to simulate the controlled-reflections.

\subsection{Universal Quantum Emulator with Ideal Controlled-Reflections (Proof of Eq.(\ref{bound42}))}
Let $\mathcal{E}_U$ be the quantum channel describing the overall effect of the circuit in Fig.(\ref{Fig}) on the system, in the case where we ignore the outcome of measurement in  step (ii), and where we use ideal controlled-reflections in the circuit. 
   
Recall that all reflections $\exp({i\pi |\phi_k^\i\rangle\langle\phi_k^\i| }): k=1,\cdots, K$  that define channel $\mathcal{D}$ are block-diagonal with respect to the subspace $\mathcal{H}^\i$ and its orthogonal subspace, and they act trivially on the orthogonal subspace. Therefore, if the support of the input state is initially restricted to $\mathcal{H}^\i$,   under the effect of channels $\mathcal{W}$ and $\mathcal{E}_U$ it remains restricted to this subspace.  Hence, in the following,  to simplify the notation and presentation,  we assume $\mathcal{H}^\i$ is equal to the entire Hilbert space,  which means states $|\phi^\i_k\rangle: k=1,\cdots, K$ span this space. Therefore, in the following we use the notion $P^\perp$ and $P^\perp_\i$  interchangeably. Also, we use the notation  $\mathcal{P}_\i$ and $\mathcal{P}$ interchangeably.\\

We start with Eq.(\ref{main_eq}) proven in Appendix \ref{App:1}, namely
\begin{align}
F(\mathcal{E}_U(\rho),U\rho U^\dag) \ge \langle\phi_1^\i|\mathcal{W}^T(\rho)|\phi_1^\i\rangle\ ,
\end{align}
where $\mathcal{W}(\sigma)=P\sigma P + \mathcal{D}(P^\perp \sigma P^\perp)$. Define 
\be
P^\perp=I-P=I- |\phi^\i_1\rangle\langle\phi_1^\i|\ ,
\ee
to be the projector to the subspace orthogonal to $|\phi^\i_1\rangle$, where $I$ is the identity operator. Let  
 \be
 \mathcal{P}^\perp(\cdot)=P^\perp(\cdot)P^\perp\ .
 \ee
  Using the definition 
  $$\mathcal{W}(\sigma)=P\sigma P + \mathcal{D}(P^\perp \sigma P^\perp)= P\sigma P+\mathcal{D}\circ \mathcal{P}^\perp(\sigma)\ , $$
  we have
  \beq\label{Eq2341}
\mathcal{P}^\perp\circ\mathcal{W}=\mathcal{P}^\perp\circ\mathcal{D}\circ\mathcal{P}^\perp=\mathcal{P}^\perp\circ\mathcal{W}\circ\mathcal{P}^\perp\ .
\eeq
This implies that for state $\rho$ restricted to $\mathcal{H}_\i$
\bes\label{Eq23}
\begin{align}
F(\mathcal{E}_U(\rho),U\rho U^\dag) &\ge \langle\phi_1^\i|\mathcal{W}^T(\rho)|\phi_1^\i\rangle\\ &=1- \Tr\left( P^\perp  \mathcal{W}^T(\rho)\right)\\   &=1- \Tr\left(\mathcal{P}^\perp\circ  \mathcal{W}^T(\rho)\right)\\ &= 1-\Tr\left(\left[\mathcal{P}^\perp\circ \mathcal{D}\circ \mathcal{P}^\perp\right]^T(\rho)\right) \ ,
\end{align}
\ees
where to get the last line we have used Eq.(\ref{Eq2341}).  We conclude that
\begin{align}
F(\mathcal{E}_U(\rho),U\rho U^\dag) &\ge \langle\phi_1^\i|\mathcal{W}^T(\rho)|\phi_1^\i\rangle=1- \Tr(\mathcal{D}_\perp^T(\rho))\ ,
\end{align}
where
$$\mathcal{D}_\perp(\cdot)=\mathcal{P}^\perp\circ \mathcal{D}\circ \mathcal{P}^\perp(\cdot)=\frac{1}{K-1} \sum_{k=2}^K A_k (\cdot) A_k\ , 
$$
with
$$ A_k=P_\perp R^\i(k) P_\perp=P_\perp-2 P_\perp |\phi^\i_k\rangle\langle\phi^\i_k| P_\perp \ ,
 $$ 
 and $P_\perp$ is the projector to $(d-1)$-dimensional subspace of $\mathcal{H}_\i$ orthogonal to $|\phi^\i_1\rangle$.

Quantum operation $\mathcal{D}_\perp$ has nice properties which simplify the following analysis. In particular, it has  Hermitian matrix representations:  Let $\{F_\mu:\mu=1\cdots d^2\}$ be an orthonormal basis for the operator space acting on $\mathcal{H}_\i$, such that $\Tr(F^\dag _\mu F_\nu)=\delta_{\mu,\nu}$. Consider the matrix representation of  $\mathcal{D}_\perp$, i.e., the matrix 
$$D_{\mu,\nu}^\perp=\Tr(F^\dag_\mu\ \mathcal{D}_\perp(F_\nu)) \ . $$ 
Then, the fact that any reflection $e^{i\pi |\phi\rangle\langle\phi |}$ is a Hermitian operator, implies that this matrix Hermitian. Diagonalizing this matrix we find
\be
\mathcal{D}_\perp(X)=\sum_\mu \lambda_\mu\  B_\mu \Tr(B^\dag_\mu X)  \ ,
\ee
where $\{\lambda_\mu\}$ are the eigenvalues of $\mathcal{D}_\perp$, which are real, and $\{B_\mu\}$ are the corresponding eigenvectors, which form an orthonormal basis, such that $\Tr(B_\mu B^\dag_\nu)=\delta_{\mu,\nu}$. Equivalently, using the vectorization of this equation, we find that 
$\{\lambda_\mu\}$ are the eigenvalues of the operator
\be
\frac{1}{K-1}\sum_{k=2}^K A_k\otimes A_k^\ast\ ,
\ee
 which is Hermitian. 

 Let $\lambda_\perp=\|\frac{1}{K-1}\sum_{k=2}^K A_k\otimes A_k^\ast\|_\infty $ be the largest eigenvalue of $ \mathcal{D}_\perp$. 
It follows that for any $T$, the map $\mathcal{D}_\perp^T$ has also a Hermitian matrix representation, and its largest eigenvalue is $\lambda^T_\perp$. Then, for any pair of operators $X$ and $Y$, it holds that
\begin{align}
\Big|\Tr\left(X^\dag\mathcal{D}_\perp^T(Y)\right)\Big|&\le \lambda_\perp^T \times \sum_{\mu} |\Tr(B_\mu X)| \times |\Tr(B^\dag_\mu Y)|\ \le \lambda_\perp^T \times \sqrt{\sum_{\mu} |\Tr(B_\mu X)|^2} \times \sqrt{\sum_{\mu} |\Tr(B_\mu Y)|^2}   \\  &= \lambda_\perp^T\ \times  \sqrt{\Tr(X^\dag X)}\sqrt{\Tr(Y^\dag Y)}\ \ ,
\end{align}
where we have used the Cauchy-Schwarz inequality.

This implies that 
\bes
\begin{align}
\Tr\left(\mathcal{D}_\perp^T(\rho)\right)&=\Tr\left(P^\perp \mathcal{D}_\perp^T(\rho)\right)  \le  \lambda_\perp^T  \times \sqrt{d-1}  \ ,
 \end{align}
 \ees
 where we have used $\Tr(\rho^2)\le 1$ and $\Tr(P^\perp)\le d-1$.
 
Putting this into Eq.(\ref{Eq23}), we find that
\begin{align}
F(\mathcal{E}_U(\rho),U\rho U^\dag)&\ge 1-\Tr\left(\mathcal{D}_\perp^T(\rho)\right)\ge 1-\sqrt{d-1}\times |\lambda_\perp|^T\ .
\end{align}
Next, we use this bound together with the Fuchs-van de Graaf inequality, to bound  the trace distance of $\mathcal{E}_U(\rho)$ and $U\rho U^\dag $. We find
\begin{align}
\big\|\mathcal{E}_U(\rho)-U\rho U^\dag \big\|_1 &\le  2\sqrt{1-F^2(\mathcal{E}_U(\rho),U\rho U^\dag)} \le 2  \sqrt{2\times (1-F(\mathcal{E}_U(\rho),U\rho U^\dag))} \\  &\le   2\sqrt{2 (d-1)\times |\lambda_\perp|^{T}}\ .
\end{align}
This means that to achieve error 
\be
\epsilon_\text{id}=\frac{1}{2}\big\|\mathcal{E}_U(\rho)-U\rho U^\dag \big\|_1\ ,
\ee
 in trace distance, it is sufficient to have 
\beq\label{fef}
T\ge  \frac{\ln (2\epsilon^{-2}_\text{id}[d-1]) }{\ln|\lambda|^{-1}_\perp} \ .
\eeq
Then, the fact that $\ln (1+x)\le x$ for $x>-1$, implies 
\be
\frac{1}{1-|\lambda_\perp|} \ge \frac{1}{\ln|\lambda_\perp|^{-1}}\ .
\ee
Therefore, to achieve error $\epsilon_\text{id}$ it is sufficient to have 
\beq\label{fef23}
T\ge  \frac{\ln (2\epsilon^{-2}_\text{id}[d-1]) }{1-\lambda_\perp} \ ,
\eeq
which is Eq.(\ref{bound42}) in the paper.

\subsection{Total error in the algorithm with simulated controlled-reflections (Proof of Eq.(\ref{Ntot}) and Eq.(\ref{Ttot}))}\label{App:error2}

As we showed above, in the idealized case where we are given perfect controlled-reflections by choosing $T$ satisfying the bound in Eq.(\ref{fef23}) we can implement the transformation $|\psi\rangle\rightarrow U|\psi\rangle$ on any initial state $|\psi\rangle\in \mathcal{H}_\i$ with error less than or equal to $\epsilon_\text{id}$ in trace distance.  But, in the main algorithm we need to simulate the controlled-reflections using the given copies of the sample states.



As shown in \cite{pichler2016measurement} and reviewed in Appendix \ref{app:SWAP}, using $n$ copies of state $|\phi\rangle$ we can implement the unitary $e^{i\pi|\phi\rangle\langle\phi|}$, or its controlled version, $V=|0\rangle\langle 0|\otimes I+ |1\rangle\langle 1| \otimes e^{i\pi|\phi\rangle\langle\phi|}$, with gate  complexity $\mathcal{O}(\log D\times n)$, and error $\mathcal{O}(1/n)$ in trace distance.  In other words, one can achieve  the error of order $\epsilon_\text{DME}$ in simulating each controlled-reflection with gate complexity  ${\mathcal{O}}(\epsilon^{-1}_\text{DME} \times \log D)$, and using ${\mathcal{O}}(\epsilon^{-1}_\text{DME})$ copies of the corresponding state. 

Using the properties of trace distance, and in particular, its monotonicity under completely-positive trace-preserving maps, and triangle inequality, one can easily show that if we replace each ideal controlled reflections with its approximate version that can be realized with density matrix exponentiation  with error $\epsilon_{\text{DME}}$, the output of the circuit will change by, at most, $\epsilon_{\text{DME}}$. Given that the circuit requires $4T+1$ controlled-reflection,  it follows that we can implement the transformation $|\psi\rangle\rightarrow U|\psi\rangle$\ , with the total error bounded by 
\beq
\epsilon_\text{tot}=[4T(d,\epsilon_\text{id})+1]\times \epsilon_\text{DME} + \epsilon_\text{id}\ ,
\eeq
with the total gate complexity 
\beq
t_\text{tot}=[4T(d,\epsilon_\text{id})+1] \times {\mathcal{O}}(\epsilon^{-1}_\text{DME} \times \log D),
\eeq
and using the total number of input-output sample pairs  
\beq
N_\text{tot}=[4T(d,\epsilon_\text{id})+1] \times \mathcal{O}(\epsilon^{-1}_\text{DME})\ .
\eeq  
Choosing $\epsilon_\text{id}$ and $ \epsilon_\text{DME} $ such that $$T(d,\epsilon_\text{id})\times \epsilon_\text{DME} =\epsilon_\text{id}=\epsilon\ , $$ we find that the transformation $|\psi\rangle\rightarrow U|\psi\rangle$ can be implemented with the total error less than or equal to $\epsilon$, with the gate complexity 
\beq
t_\text{tot}=\mathcal{O}(T^2\times \epsilon^{-1}) \times \log D=\widetilde{\mathcal{O}}\left(\frac{d^2 \epsilon^{-1}  \log D  }{(1-\lambda_\perp)^2 }\right) ,
\eeq
and using the total number of input-output sample pairs  
\beq
N_\text{tot}=\mathcal{O}(T\times \frac{1}{\epsilon/T})= \widetilde{\mathcal{O}}\left(\frac{d^2 \epsilon^{-1}  }{(1-\lambda_\perp)^2 }\right)  \ ,
\eeq  
where $\widetilde{\mathcal{O}}$  suppresses more slowly-growing terms, i.e., multiplicative terms in $\log  \epsilon^{-1}$ and $\log d$.

 \color{black}

\newpage

\section{Quantum circuit for exponentiating a density operator}\label{app:SWAP}

In this section,  we  review   the density matrix exponentiation technique introduced in \cite{lloyd2014quantum}  and 
then discuss the complexity of its implementation via a method used in \cite{pichler2016measurement}.

Let $S$ be the SWAP operator acting on the system  and an ancillary system with equal Hilbert space dimensions, such that $S|\mu\rangle|\nu\rangle=|\nu\rangle|\mu\rangle$, for any pair of states $|\mu\rangle$, and $|\nu\rangle$.  Note that $S$ is a Hermitian unitary operator. Then, for $|\theta|\ll 1$
\beq\label{suz}
\Tr_\text{RF}\left(e^{-i \theta S} [\rho\otimes \sigma] e^{i \theta S}\right)= \rho+i \theta [\rho,\sigma] +\mathcal{O}(\theta^2)=e^{-i \sigma\theta} \rho e^{i \sigma  \theta}+\mathcal{O}(\theta^2)  \ ,
\eeq
where the partial trace is over the system with state $\sigma$ (which can be interpreted as a  quantum reference frame \cite{marvian2008building, QRF_BRS_07}). More precisely, the error, as quantified by the trace distance between the state in the left-hand side and the desired state $e^{-i \sigma\theta} \rho e^{i \sigma  \theta}$ is bounded by $\mathcal{O}(\theta^2)$ \cite{lloyd2014quantum, kimmel2017hamiltonian}.

  As we explain in the following,  the unitary $e^{-i \theta S} $ for arbitrary $\theta$, can be realized using $\mathcal{O}(\log{(D)})$ elementary gates, where $D$ is the dimension of the Hilbert space.   In summary,  using one copy of state $\sigma$ we can implement  the unitary $e^{-i \sigma \theta}$ with the error of order $\mathcal{O}(\theta^2)$, and with   $\mathcal{O}(\log{(D)})$ elementary gates.  Then, having $n$ copies of state $\sigma$, to implement  the unitary $e^{-i \sigma t} $ we can choose $\theta=t/n$ and use each copy to simulate  $e^{-i \sigma \theta}$.  In this case the gate complexity is
\beq\label{t23}
t_\text{tot}= \mathcal{O}( \log(D) \times  n) \ ,
\eeq
and  the total error is
\beq\label{er23}
\epsilon_\text{tot}= \mathcal{O}\big(n\times  \theta^2\big)=\mathcal{O}\big(\frac{t^2}{n} \big)\ . 
\eeq
That is
\be
\Big\|\mathcal{F}_n(\cdot)- e^{-i t\sigma} (\cdot)e^{i t\sigma} \Big\|_\diamond =\mathcal{O}(\frac{t^2}{n})\ ,
\ee
where $\|\cdot \|_\diamond $ denotes the diamond norm distance, and 
\be
 \mathcal{F}_n(\rho)= \Tr_{2,\cdots, n+1}\Big(e^{-i S_{1,n+1} t/n} \cdots e^{-i S_{1,2} t/n} )[\rho\otimes \sigma^{\otimes n}] (e^{i S_{1,2} t/n} \cdots e^{i S_{1,n+1} t/n} \Big)\ ,
\ee
is the quantum channel describing the above process. Here $S_{1, j}$ is the swap operator  on the main system, labeled as 1, and the ancillary system, labeled as $j=2,\cdots, n+1$.

To implement the controlled reflections that are used in the universal quantum emulator, 
we consider the controlled version of Eq.(\ref{suz}), namely
\beq
\Tr_\text{RF}\Big(\exp(-i\theta S_c)[\rho\otimes \sigma] \exp(i\theta S_c)\Big)= |0\rangle\langle 0|\otimes \rho+ |1\rangle\langle 1|\otimes  \exp({-i \sigma \theta}) \rho \exp({i \sigma \theta})+\mathcal{O}(\theta^2)  \ ,
\eeq
where 
\be\label{contS}
\exp(i\theta S_c)=\exp\big(i\theta [|0\rangle\langle 0|\otimes I+|1\rangle\langle 1|\otimes S] \big) =\exp({i\theta}) |0\rangle\langle 0|\otimes I+  |1\rangle\langle 1|\otimes \exp({i\theta S})\ .
\ee   
Then, a similar argument implies that using $n$ copies of state $\sigma$ we can implement  the controlled unitary
\be
|0\rangle\langle 0|\otimes I+ |1\rangle\langle 1|\otimes  \exp(i t \sigma)\ ,
\ee
with error given in Eq.(\ref{er23}) and gate complexity given  in   Eq.(\ref{t23}). 

Next, we explain how the unitaries $\exp(i\theta S)$ and $\exp(i\theta S_c)$ can be implemented. Indeed, we consider a slightly more general task: Suppose $V$ is a general Hermitian unitary. Let 
\beq
V^\text{cont}=|0\rangle\langle 0|_\text{anc}\otimes I+ |1\rangle\langle 1|_\text{anc}\otimes V\ ,
\eeq 
be the controlled version of $V$, which is also Hermitian and unitary.  
Then, consider the identity
\begin{align}
V^\text{cont}\ \exp(i\theta  X_\text{anc}) \ V^\text{cont} &=  \exp(i\theta  V^\text{cont} X_\text{anc} V^\text{cont})  =\exp(i\theta X_\text{anc}\otimes V)\ ,
\end{align}
which corresponds to the circuit that first implements the controlled version of $V$, then applies the single-qubit rotation $ \exp(i{\theta} X_\text{anc})$ on the ancilla qubit, and finally applies the controlled version of $V$ again.  Applying this circuit on the initial state $|+\rangle\otimes |\psi\rangle$, where $|\psi\rangle$ is an arbitrary state of the system,  we obtain state
\beq
[V^\text{cont}\ \exp(i\theta X_\text{anc}) \ V^\text{cont}](|+\rangle\otimes |\psi\rangle)=\exp(i\theta X_\text{anc}\otimes V) (|+\rangle\otimes |\psi\rangle)=|+\rangle\otimes \exp({i\theta V})|\psi\rangle\ .
\eeq 
Therefore, at the end of the circuit ancilla remains in state $|+\rangle$, whereas the system evolves according to unitary  $\exp({i\theta V})$. In conclusion,  ignoring the single gate 
 $\exp(i{\theta} X_\text{anc})$,  the complexity of implementing $V$ as a Hamiltonian is, at most, twice the 
 complexity of implementing the controlled version of $V$ as a unitary.

To realize reflections $R(i\pi |\phi\rangle\langle\phi|)$ via  density matrix exponentiation, we assume the unitary $V=S$ is the swap unitary on $\mathbb{C}^D\otimes \mathbb{C}^D$. Then, the above circuit requires controlled swap, $S_c$.   Assuming each copy of the system contains $n$ qubits and hence has the total Hilbert space of dimension $D=2^n$, $S_c$ can be realized with  $n=\log D$ 2-qubit controlled-SWAPs, also known as the Fredkin gate. 

Similarly, if we choose $V$ to be the operator $S_c$, then the unitary $\exp(i\theta S_c)$ in Eq.(\ref{contS}) can be realized using two copies of controlled-controlled swaps, each of which requires $n=\log D$  controlled Fredkin gates. 

In summary, we conclude that using $n$ copies of state $\sigma$, we can realize the unitary $ \exp(i t \sigma)$ and its controlled version
$|0\rangle\langle 0|\otimes I+ |1\rangle\langle 1|\otimes  \exp(i t \sigma)$ with error given in Eq.(\ref{er23}) and the total number of gates given in Eq.(\ref{t23}).

\newpage

\section{A Modified Circuit with Exponentially Less Ancillary Qubits}\label{App:Sec:comp}

\begin{figure*}
  \includegraphics[scale=1]{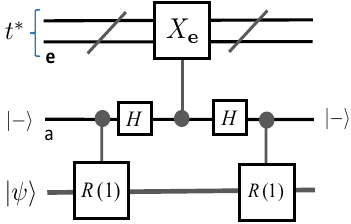}
\caption{Unitaries $V^{\i,\o} $ are defined by the above circuit. The unitary $\textbf{X}_\textbf{e}$ applies a cyclic shift by +1, in the computational basis, i.e., $\textbf{X}_\textbf{e}|z\rangle_\textbf{e}=|z+1 (\text{mod } 2^{t^\ast})\rangle_\textbf{e}.$ The unitaries $R^{\i/\o}(1)$ are reflections about  $|\phi_1^{\i/\o}\rangle$. Note that we have suppressed the $\i$ and $\o$ superscripts in the above circuit.  
   \label{Fig4}}
\end{figure*}

In this section, we present a modified version of the circuit in Fig.\ref{Fig} that generates exactly the same output using only $\lceil \log (T+1)\rceil$ ancillary qubits, which is an exponential improvement 
 over the original circuit that uses $T$ ancillary qubits. This modified circuit has exactly the same performance in terms of the sample complexity and error, whilst its runtime is increased only by a multiplicative factor of $\log T$.   The idea is simple: The circuit in Fig.\ref{Fig} uses a unary encoding to encode information in the ancilla qubits. We replace this with a binary encoding.  

As stated in Eq.(\ref{Eq_expan}) and Eq.(\ref{rec2}), at the end of step (i) of the algorithm the joint state of the main system and ancilla qubits is given by
\beq\label{ap11}
|\phi^\i_1\rangle\otimes \big(\sum_{t=0}^{T-1}  \langle\phi^\i_1|\psi(t,\textbf{k})\rangle \  |t\rangle_\textbf{a}\big) +|\psi(T,\textbf{k})\rangle\otimes |T\rangle_\textbf{a}  \ ,
\eeq
where the (unnormalized) vectors $|\psi(t,\textbf{k})\rangle$ are defined via the recursive relation 
\beq
|\psi(t+1,\textbf{k})\rangle=R^\i(k_{t+1}) P^\perp |\psi(t,\textbf{k})\rangle \ ,
\eeq
and $|\psi(0,\textbf{k})\rangle=|\psi\rangle$. 
 Here, $|t\rangle_\textbf{a}=|1\rangle^{\otimes t}\otimes  |0\rangle^{\otimes (T-t)}$ is  the state of qubits $a_1\cdots a_T$, in which $a_{t+1} \cdots a_{T}$ are all in state $|0\rangle$, and the rest of qubits are in state $|1\rangle$. 
 
 If instead of this unary encoding we use the binary encoding, then we can represent  these $T+1$ orthonormal states using only 
 \be
 t^\ast \ge \log (T+1) 
 \ee
  ancillary qudits. Namely, we replace  
 \be
 |t\rangle_\textbf{a}=|1\rangle^{\otimes t}\otimes  |0\rangle^{\otimes (T-t)} \ \ \longrightarrow \ \ |t\rangle_\textbf{e}=|t_{t^\ast}\rangle\otimes \cdots \otimes |t_1\rangle\ ,
 \ee
 where $\textbf{e}$ denotes $ t^\ast$ ancillary qubits, and $t_{t^\ast}\cdots t_1$ is the binary expansion of $t$. 
 
 Then, we can easily modify the steps in the original circuit in Fig.\ref{Fig} to generate the binary encoded version of state in Eq.(\ref{ap11}), i.e., state 
$|\phi^\i_1\rangle\otimes \big(\sum_{t=0}^{T-1}  \langle\phi^\i_1|\psi(t,\textbf{k})\rangle \  |t\rangle_\textbf{e}\big) +|\psi(T,\textbf{k})\rangle\otimes |T\rangle_\textbf{e} $.

The modified circuit has two main ingredients. First, as before, we need to perform  the projective measurement with projectors $\{P=|\phi^\i_1\rangle\langle\phi_1^\i, P^\perp=I-P\}$. This is done using the circuit $V^\i$ presented in Fig. \ref{Fig4}. In this circuit we use the controlled unitary $|0\rangle\langle 0|_1\otimes I_\textbf{e}+|1\rangle\langle 1|_1\otimes \textbf{X}_\textbf{e}$, where $\textbf{X}_\textbf{e}$ cyclically shifts the state of qubits  $\textbf{e}$ by +1, i.e., 
\be
\textbf{X}_\textbf{e}=\sum_{t=0}^{2^{t^\ast}-2} |t+1\rangle\langle t|_\textbf{e}+|0\rangle\langle 2^{t^\ast}-1|_\textbf{e}=\sum_{t=0}^{2^{t^\ast}-1}   |t+1\ (\text{mod } 2^{t^\ast})\rangle\langle t|_\textbf{e}\ .
\ee 
Note that $\textbf{X}_\textbf{e}$ and its controlled version  $|0\rangle\langle 0|_1\otimes I_\textbf{e}+|1\rangle\langle 1|_1\otimes \textbf{X}_\textbf{e}$ 
can be implemented in time $\Theta(\log(T))$.

The circuit $V^\i$ presented in Fig.\ref{Fig4} coherently performs the projective  measurement with projectors $\{P=|\phi^\i_1\rangle\langle\phi_1^\i, P^\perp=I-P\}$, and if the system is in state $|\phi^\i_1\rangle$, it does nothing, whereas if it is in a state   orthogonal to $|\phi^\i_1\rangle$, it applies a cyclic shift on the state of qubits $\textbf{e}$, i.e., 
\be
V^\i (|-\rangle_a|t\rangle_\textbf{e}|\psi\rangle)=|-\rangle_a\big(\langle\phi^\i_1|\psi\rangle\ |t\rangle_\textbf{e}|\phi^\i_1\rangle+ |t+1\rangle_\textbf{e}P^\perp|\psi\rangle\big)\ .
\ee 
In addition to this coherent measurement, we also need to apply an encoded version of the controlled-reflections. Let $F^{\text{in/out}}_\textbf{e}(t,k_t)$ be the controlled-reflections controlled by qubits $\textbf{e}$, defined by
\be
F^{\text{in/out}}_\textbf{e}(t,k_t)=(I_\textbf{e}-|t\rangle\langle t|_\textbf{e})\otimes I_S+|t\rangle\langle t|_\textbf{e}\otimes R^{\i/\o}(k_t)\ .
\ee
This means that if the ancillary qubits $\textbf{e}$ are in states $|t\rangle_\textbf{e}$,  this gate applies the reflection $R^{\i/\o}(k_t)$ to the system, and otherwise it leaves the system unchanged.

Then, it can be easily seen that for any input state $|\psi\rangle$ of the system, by applying  
$$\prod_{t=1}^T \big[F^\i(t,k_t) V^\i\big]$$
to the system and ancillary qubits, we get
\begin{align}
\prod_{t=1}^T \big[F^\i(t,k_t) V^\i\big] (|-\rangle_a|t\rangle_\textbf{e}|\psi\rangle)&=\big[F^\i(T,k_T)V^\i\cdots  F^\i(1,k_1)V^\i\big](|-\rangle_a|t\rangle_\textbf{e}|\psi\rangle)\\ &=|\phi^\i_1\rangle\otimes  |0\rangle_a\otimes \Big[\sum_{t=0}^{T-1}  \langle\phi^\i_1|\psi(t,\textbf{k})\rangle  |t\rangle_\textbf{e}\big) +|\psi(T,\textbf{k})\rangle |T\rangle_\textbf{e} \Big] \ ,
\end{align}
which yields the desired state in Eq.(\ref{ap11}), up to a change of encoding $ |t\rangle_\textbf{a}\rightarrow  |t\rangle_\textbf{e}$.

Note that unitary $\prod_{t=1}^T \big[F^\i(t,k_t) V^\i\big] $, which implements step (i) of the algorithm is indeed compressing the input state in a $D$-dimensional Hilbert space, in $\mathcal{O}(\log(d))$ ancillary qubits. We conclude that, having multiple copies of states in an unknown $d$-dimensional subspace of the Hilbert space of $\mathcal{O}(\log D)$ qubits, we can compress any input state $|\psi\rangle$ in  $t^\ast=\lceil \log(d+1)\rceil$ ancillary qubits, in a reversible fashion with an error which is exponentially small in $T=2^{t^\ast}$, and therefore doubly exponentially small in $t^\ast$. This method for state compression could be of independent interest.

\newpage

 \section{A Physical Interpretation of Density Matrix Exponentiation}\label{App:DME2}

In this section, we discuss a slightly different point of view to Density Matrix Exponentiation. 
Consider $n$ ancillary systems,  also known as Quantum Reference Frames (QRF), each of which has a Hilbert space isomorphic to the Hilbert space of the main system.    
Suppose initially the main system is in state $\rho$ and each ancillary system is in the pure state $|\phi\rangle$, so that their joint state is 
$$\rho\otimes |\phi\rangle\langle\phi|^{\otimes n}\ .$$ 
Suppose at time $t=0$ we couple the main system to these ancillary systems via the Hamiltonian
\beq
H=J\sum_{j=1}^{n} S_{\text{sys},j}\ ,
\eeq
where  $S_{\text{sys},j}$ is the SWAP operator  acting on the main system and $j$'th QRF, tensor product with the identity operators on the rest of systems.  Then, at time $t\ge 0$ the reduced state of the main system is given by 
\beq
\rho(t)= \Tr_\text{RF}\Big( e^{- i H t}(\rho \otimes |\phi\rangle\langle \phi|^{\otimes n}) e^{ i H t} \Big)\ ,
\eeq
where the partial trace is over all $n$ QRF.  Then, it turns out that in the limit of large $n$ the system remains almost uncorrelated with QRF's, and its evolution can be approximated by a unitary generated by the effective Hamiltonian
\beq
H_\text{eff}=J (n+1) |\phi\rangle\langle\phi |\ . 
\eeq
More precisely, by diagonalizing Hamiltonian $H$ and finding the exact reduced $\rho(t)$ we show that 
\begin{proposition}\label{Thm}
For states $\rho(t)$ and  $e^{-i t H_\text{eff} } \rho e^{i t H_\text{eff}}$  the fidelity is lower bounded as
\be\label{trw}
\text{Fid}(\rho(t), e^{-i t H_\text{eff} } \rho e^{i t H_\text{eff}}) \ge \frac{n-1}{n+1}\ ,
\ee
and the trace distance  is upper bounded by
\beq\label{first}
\frac{1}{2}\left\|\rho(t)-e^{-i t H_\text{eff}} \rho e^{i t H_\text{eff} }\right\|_1 \le \frac{6}{n}\ .
\eeq
Furthermore, for any pure state $|\psi\rangle$ of the main system, the fidelity of the joint state of the system and QRF's with  state $|\Psi_\text{ideal}(t)\rangle=e^{-i H_\text{eff} t}|\psi\rangle\otimes |\phi\rangle^{\otimes n}$ is bounded by
\beq\label{Eq:sec}
\left|\langle\Psi_\text{ideal}(t) |e^{-it H} |\psi\rangle\otimes |\phi\rangle^{\otimes n}\right|\ge \frac{n-1}{n+1}\ .
\eeq
\end{proposition}
Therefore,  in the limit of large $n$ the state of QRF's remains unaffected in this process. 

We note that a similar result has been previously established in  \cite{marvian2008building}  in the context of quantum reference frames.


\subsection*{Proof of proposition \ref{Thm}}

We prove these results for the case where $\rho$ the initial state of system is pure. The results easily generalize to mixed states using triangle inequality of trace distance and  joint concavity of fidelity.  As we will see in the following, under this assumption, the dynamics of each system is restricted to a 2d subspace, and we can treat it as a qubit. Then, one can take advantage of the SU(2) symmetry of SWAP Hamiltonian and apply the representation-theoretic techniques used previously in \cite{marvian2008building} to study the time  
evolution. However, the following analysis does not rely on such techniques. 

 An arbitrary pure state of system can be decomposed as
\beq
|\psi\rangle=\alpha|\phi\rangle+\beta |\phi^\perp\rangle\ ,
\eeq
where $|\alpha|^2+|\beta|^2=1$, and $\langle\phi|\phi^\perp\rangle=0$. 
Using this notation, the initial joint state of system and QRF's at $t=0$ is given by
\beq
|\psi\rangle |\phi\rangle^{\otimes n}=\alpha|\phi\rangle^{\otimes (n+1)}+\beta |\phi^\perp\rangle\otimes  |\phi\rangle^{\otimes n},
\eeq
State $$|v\rangle=|\phi\rangle^{\otimes (n+1)}$$ is an eigenstate of $\sum_{i=1}^n S_{\text{sys},i}$ and, hence an eigenstate of $H$:
\beq
H |\phi\rangle^{\otimes (n+1)}=J\sum_{i=1}^n S_{\text{sys},i} |\phi\rangle^{\otimes (n+1)}=J n  |\phi\rangle^{\otimes (n+1)}\ .
\eeq
Next, consider $n+1$ orthonormal states 
\bes
\begin{align}
|\bar{0}\rangle&\equiv |\phi^\perp\rangle\otimes  |\phi\rangle^{\otimes n}\\ 
|\bar{i}\rangle&\equiv S_{\text{sys},i} |\bar{0}\rangle=S_{\text{sys},i} |\phi^\perp\rangle\otimes  |\phi\rangle^{\otimes n}\ ,\ \ \ \   1 \le i \le n\ .
\end{align}
\ees
These states 
span a $n+1$-dimensional subspace, denoted by $\mathcal{H}_1$. It can be easily seen that the operator $\sum_{i=1}^n S_{\text{sys},i} $,   and hence the Hamiltonian $H$, is block-diagonal with respect to the subspace $\mathcal{H}_1$. Diagonalizing $\sum_{i=1}^n S_{\text{sys},i} $ in this subspace, we find that it has three distinct  eigenvalues $-1,n-1$  and $n$ with the corresponding eigenstates 
\bes
\begin{align}
\lambda= n &:\ \ \ |u_+\rangle= \frac{1}{\sqrt{n+1}} \sum_{j=0}^n  |\bar{j}\rangle\ \    &&(\text{non-degenerate in } \mathcal{H}_1)\\
\lambda=-1 &:  \ \ \    |u_-\rangle= \sqrt{\frac{n}{n+1}} |\bar{0}\rangle-  \frac{1}{\sqrt{n(n+1)}} \sum_{j=1}^n  |\bar{j}\rangle  \ \    &&(\text{non-degenerate in } \mathcal{H}_1)   \\
\lambda=n-1 &: \ \ \    |u_k\rangle= \frac{1}{\sqrt{n}}\sum_{j=1}^n  e^{i  2\pi jk/n}|\bar{j}\rangle\ , \ \ \ \ \  1\le k\le n-1  &&(\text{$n-1$ fold degenerate in }  \mathcal{H}_1) 
\end{align}
\ees
 State $|\psi\rangle\otimes |\phi\rangle^{\otimes n}$, the joint initial state of the system and QRFs, lives in the subspace spanned by three eigenstates of $H$, i.e. $|v\rangle, |u_+\rangle$ and ,$|u_-\rangle$ with the corresponding eigenvalues $nJ, nJ$ and $-J$. 
It follows that the joint state at time $t$ is given by
\begin{align}
|\Psi(t)\rangle\equiv (e^{-i H t} |\psi\rangle)|\phi\rangle^{\otimes n} &=(e^{-i H t}\otimes I^{\otimes n}) (\alpha|v\rangle+\beta |\bar{0}\rangle)=\alpha e^{- i t J n } |v\rangle+\beta e^{- i t J n }\frac{1}{\sqrt{n+1}}|u_+\rangle+\beta e^{i t J} \sqrt{ \frac{n}{{n+1}}} |u_-\rangle \ .
\end{align}
Next consider
\beq
|\Psi_\text{ideal}(t)\rangle\equiv e^{- i H_\text{eff} t}|\psi\rangle\otimes |\phi\rangle^{\otimes n}=\alpha e^{-i J (n+1) t} |\phi\rangle^{\otimes (n+1)}+\beta |\phi^\perp\rangle\otimes  |\phi\rangle^{\otimes n}=\alpha e^{-i J (n+1) t}|v\rangle+\beta |\bar{0}\rangle\ . 
\eeq
Fidelity of this state with state $|\Psi(t)\rangle=e^{-i H t} |\psi\rangle|\phi\rangle^{\otimes n}$ is given by
\beq
\left|\langle\Psi(t)|\Psi_\text{ideal}(t)\rangle\right|\ge \frac{n-1}{n+1}\ ,
\eeq
which also implies Eq.(\ref{trw}). 

Next, we focus on the reduced state of system at time $t$. It can be easily seen that at any time $t$ the reduced state of system is restricted to the two dimensional subspace spanned by $|0_\s \rangle= |\phi\rangle$ and $|1_\s \rangle= |\phi^\perp_\RF\rangle$. Furthermore, the reduced state of the QRF's is also restricted to a two dimensional subspace spanned by 
\begin{align}
|0_\RF\rangle&\equiv|\phi\rangle^{\otimes n}\\
|1_\RF\rangle&\equiv\frac{1}{\sqrt{n }}\sum_{k=0}^{n-1} |\phi\rangle^{\otimes k}\otimes |\phi^\perp_\RF\rangle\otimes  |\phi\rangle^{\otimes n-k-1} \ .
\end{align}
Then, at any time $t\ge 0$ the joint state of system and QRF's can be written as 
\begin{align}
|\Psi(t)\rangle\equiv e^{-i H t} |\psi\rangle|\phi\rangle^{\otimes n} =e^{i J t}\left[\left(\alpha e^{-i t J (n+1)} |0_\s\rangle+\beta \frac{n+e^{-it J (n+1)}}{n+1} |1_\s\rangle \right)\otimes |0_\RF\rangle+\beta (e^{-it J (n+1)}-1)\frac{\sqrt{n}}{n+1} |0_\s\rangle \otimes |1_\RF\rangle\right]\ .
\end{align}
It follows that at any time $t$ the reduced density operator of system in the basis $\{|\phi\rangle,|\phi^\perp\rangle\}$  is given by
\begin{align}
\rho(t)=\left(
\begin{array}{cc}
|\alpha|^2+|\beta|^2 \times n |\delta|^2&\ \   \alpha\beta^\ast e^{-i (n+1)J t}\times (1-\delta) \\ \ \\
 \alpha^\ast\beta e^{i (n+1)J t}\times (1-\delta^\ast)  &  |\beta|^2 (1-n |\delta|^2)  
\end{array}
\right)\ ,
\end{align}
where $\delta=[1-e^{-i (n+1)Jt}](n+1)^{-1}$, and so  $|\delta|\le\frac{2}{n+1}$. Note that the ideal state $e^{-it H_\text{eff}}|\psi\rangle$ in this basis is 
\begin{align}
e^{-it H_\text{eff}}|\psi\rangle\langle\psi| e^{it H_\text{eff}} =\left(
\begin{array}{cc}
|\alpha|^2&\ \   \alpha\beta^\ast e^{-i (n+1)J t} \\ \ \\
 \alpha^\ast\beta e^{i (n+1)J t} &  |\beta|^2  
\end{array}
\right)\ .
\end{align}
Therefore,
\begin{align}
\left\|\rho(t)-e^{-it H_\text{eff}}|\psi\rangle\langle\psi| e^{it H_\text{eff}} \right\|_1\le 2n|\delta|^2+2|\delta|\le \frac{12}{n}\ .
\end{align}

\newpage
\color{black}

\section{Approximate transformations (Proof of Eq.(\ref{bound4}))}\label{App:approx}

In the following we show Eq.(\ref{bound4}):
\be\nonumber
\frac{1}{2}\left\|\mathcal{E}_U(\rho)-\widetilde{\mathcal{E}}(\rho)\right\|_1\le T\times 4\delta \ .
\ee
First, note that 
$$\left\|\exp({i \pi |\phi\rangle\langle\phi|})-\exp({i \pi |\widetilde{\phi}\rangle\langle\widetilde{\phi}|})\right\|_\infty=2\sqrt{1 -|\langle\phi|\widetilde{\phi}\rangle|^2} ,$$
where $\|\cdot \|_\infty$ denotes the operator norm, and 
we have used the fact that $\exp({i \pi |\phi\rangle\langle\phi|})=I-2|\phi\rangle\langle\phi|$.  By Eq.(\ref{bound1}), $|\langle\widetilde{\phi}^\o_k| U|\phi_k^\i\rangle|^2\ge 1-\delta^2$, which implies
\be
\big\| R^\o({k})_a-\widetilde{R}^\o({k})_a\big\|_\infty\le 2\delta\ ,
 \ee
 where  ${R}^\o({k})_a$ and $\widetilde{R}^\o({k})_a$ respectively  denote controlled-reflections with respect to  $|{\phi}^\o_k\rangle=U|{\phi}^\i_k\rangle$ and $|\widetilde{\phi}^\o_k\rangle$. This, in turn,  implies 
\be
\|{W}_a^\o(k)-\widetilde{W}_a^\o(k)\|_\infty \le 4\delta \ ,
\ee
where $\widetilde{W}_a^\o(k)=\widetilde{R}_a^\o(k) H_a \widetilde{R}_a^\o(1)$, and we have used the fact that for any unitaries $V_1, V_2, \widetilde{V}_1, \widetilde{V}_2 $, it holds that 
\be
\|\widetilde{V}_2\widetilde{V}_1-V_2V_1\|_\infty\le  \|\widetilde{V}_2-V_2\|_\infty 
+\|\widetilde{V}_1-V_1\|_\infty\ . 
\ee 
This means that if in the circuit in Fig.\ref{Fig}, we replace ${W}_a^\o(k)=(U\otimes I_a)W_a^\i(k) (U^\dag\otimes I_a)$  with $\widetilde{W}_a^\o(k)$, we realize a unitary whose distance with the original unitary is bounded by $T \times 4\delta $. Finally, we use the fact that for any
unitary $V$ and state $|\psi\rangle$, it holds that
 $$\|V |\psi\rangle\langle\psi|V^\dag-|\psi\rangle\langle\psi|\|_1\le \|V |\psi\rangle\langle\psi|V^\dag-V|\psi\rangle\langle\psi|\|_1+ \| V|\psi\rangle\langle\psi|-|\psi\rangle\langle\psi|\|_1\le 2\|V-I\|_\infty \ .$$
This in turn implies that for any 
 density operator $\sigma$ and pair of unitaries $V$ and 
$\widetilde{V} $, it holds that 
\be
\frac{1}{2}\|V\sigma V^\dag-\widetilde{V} \sigma \widetilde{V}^\dag  \|_1\le \|V-\widetilde{V}\|_\infty\ ,
\ee
which means for any probability distribution $\{p_j\}$ and set of unitaries  $\{V_j\}$ and $\{\widetilde{V}_j\}$, it holds that
 \be
\frac{1}{2}\big\| \sum_j p_j V_j\sigma V_j^\dag- \sum_j p_j  \widetilde{V}_j \sigma \widetilde{V}^\dag_j  \big\|_1\le \sum_j p_j  \|V_j-\widetilde{V}_j\|_\infty\ .
\ee
It follows that if in the circuit in Fig.\ref{Fig}, we replace ${W}_a^\o(k)=(U\otimes I_a)W_a^\i(k) (U^\dag\otimes I_a)$  with $\widetilde{W}_a^\o(k)$, the resulting joint state of the system and ancillary qubits will have trace distance with the original state bounded by $4\delta T$. Finally,  by tracing over the ancilla qubits, and using the monotonicity  of the  trace distance under partial trace, we find $\frac{1}{2}\left\|\mathcal{E}_U(\rho)-\widetilde{\mathcal{E}}(\rho)\right\|_1\le T\times 4\delta$, which is Eq.(\ref{bound4}).

\color{black}

\newpage

\section{Generalizations}\label{App:Gen}
In this section, we discuss several generalizations of the results presented in the paper. These generalizations are summarized in Sec.\ref{Sec:Disc} of the paper.

\subsection{Input-output pairs of mixed states}\label{Sec:mixed}

The algorithm presented in the paper assumes the given samples of the input-output pairs are all pure states. However, as we explained at the end of Sec.\ref{Sec:alg} this algorithm can be generalized to the case where the  sample states contain mixed states. It can be easily shown that, as long as the sample input states $S_\i=\{\rho^\i_k, k=1,\cdots K\}$  contain (at least) one state close to a pure state, we can still coherently erase the sate of system and push it into this  pure state.  Then, we can emulate the action of the unknown unitary $U$, using the same approach we used in the original algorithm.

 Let $\mathcal{H}_\i$ be the subspace spanned by the union of the support of density operators in $S_\i=\{\rho^\i_k, k=1,\cdots K\}$. Then, having copies of sample input-output states in $S_\i$ and $S_\o$, we can determine the action of $U$ on any state in  $\mathcal{H}_\i$,  if and only if the set $S_\i$ generates the full matrix algebra on $\mathcal{H}_\i$. Therefore, in the following we naturally assume this condition is satisfied. 

To erase the state of system coherently we use controlled-translations  with respect to the given sample states, i.e. the unitaries
\beq
T_{a}(k,t)=|0\rangle\langle0|_a\otimes I+|1\rangle\langle1|_a\otimes e^{-i\rho_k t}\ ,
\eeq 
where we have suppressed  \emph{in} and \emph{out} superscript in both sides. As we discussed in Appendix \ref{app:SWAP}, using  the given copies of sample states we can efficiently simulate these unitaries. Recall that the input set $S_\i$ (and  hence $S_\o$) contains at least one pure state. Without loss of generality, let $\rho_1^\i$ and $\rho_1^\o=U\rho_1^\i U^\dag$ be the pure state in the  sample set $S_\i$ and  its corresponding output state in the set $S_\o$, respectively. Then instead of unitaries $W_a(k)=R_{a}(k) H_{a} R_{a}(1)$ used in the original algorithm, we use unitaries $W_a(k,t)=T_{a}(k,t) H_{a} T_{a}(1,\pi)$, where we have suppressed \emph{in} and \emph{out} superscripts, and choose $k$  and $t$ uniformly at random from the sets $1,\cdots ,K$ and $[0,1]$, respectively.  Then, it can be easily shown that state $\rho_\i^1=|\phi^\i_1\rangle\langle\phi^\i_1|$ is the unique fixed point state of  channel $\mathcal{W}=1/K\sum_{k=1}^K \int_0^1 dt \Tr_a\big(W_a^\i(k,t) [\tau\otimes |-\rangle\langle-|_a] {W_a^\i}^\dag(k,t) \big)$ inside $\mathcal{H}_\i$. It follows that we can coherently erase the state of system, and therefore, using the same technique we used in the main algorithm, we can emulate the action of unitary $U$.

\subsection{Emulating controlled unitaries}
In many quantum algorithms, such as quantum phase estimation, one needs to implement the controlled version of a unitary $U$, i.e. the unitary $U^c\equiv|0\rangle\langle 0|\otimes I+|1\rangle\langle 1|\otimes U$. Can we modify the proposed algorithm to implement the controlled version of $U$ as well?  

To answer this question, first note that if the only given resources are multiple copies of  states in $S_\i$ and $S_\o$, it is impossible to distinguish between  unitaries $U$ and  $e^{i\theta} U$, for any phase $e^{i\theta}$. On the other hand, in general the controlled version of  $U$ and  $e^{i\theta} U$ are distinct unitaries, which are not equivalent up to a global phase. This means that even to specify what is the controlled version of $U$ we need extra resources that define and fix this global phase. For instance, we can use multiple copies of the input state 
$|\Phi^\i\rangle=(\alpha|0\rangle+ \beta |1\rangle)\otimes |\phi^\i\rangle\ ,$ 
with $\alpha\neq 0$ and $\beta\neq 0$ and $|\phi^\i\rangle\in\mathcal{H}_\i$, together with  copies of its corresponding output $|\Phi^\o\rangle=U^c |\Phi^\i\rangle$.

Therefore, in the following we assume in addition to the multiple copies of states in $S_\i$ and $S_\o$, we are also given multiple copies of states $|\Phi^\i\rangle$ and $|\Phi^\o\rangle$. Again, we naturally assume the set $S_\i$ generates the full matrix algebra on $\mathcal{H}_\i$. This together with the fact that $|\phi^\i\rangle\in\mathcal{H}_\i$ implies  the set
\beq
 \{|0\rangle\langle 0 |, |1\rangle\langle 1|\}\otimes S_\i\ \  \cup \{|\Phi^\i\rangle\langle \Phi^\i|\}
\eeq 
generates the full matrix algebra on $\mathbb{C}^2\otimes \mathcal{H}_\i$, where $\mathbb{C}^2$ denotes the Hilbert space of the controlled qubit. It follows that, given these resources we can now implement the algorithm proposed in the paper to emulate the controlled-unitary $U^c$ on $\mathbb{C}^2\otimes \mathcal{H}_\i$. In other words, if instead of states in the sets $S_\i$ and $S_\o$ we choose  states  from the sets $\big(\{|0\rangle\langle 0 |, |1\rangle\langle 1\}\otimes S_\i\big)\cup\{|\Phi^\i\rangle\langle \Phi^\i|\}$ and $\big(\{|0\rangle\langle 0 |, |1\rangle\langle 1\}\otimes S_\o\big)\cup\{|\Phi^\o\rangle\langle \Phi^\o|\}$ respectively, we implement unitary $U^c$.


\subsection{More efficient algorithm with prior information about samples}
The algorithm presented in the paper does not assume any prior information about the sample input states, or the relation between them. On the other hand, as we show in the following, making some assumptions about  the sample input states, we can emulate the unknown unitary more efficiently. Note that we do not assume any prior information about any single sample states; the assumption is only about the relation between them. More precisely, the assumption is about the pairwise inner product between the states in the input set.

 The main idea behind this version of the algorithm is again emulating via coherent erasing. We use the fact that by measuring the system in two conjugate bases, we can completely erase the state of system. Therefore, 
we assume we are given multiple copies of states in an orthonormal basis for the input subspace, and their corresponding output states. We also assume we are given multiple copies of one (or more) state in the conjugate basis, and their corresponding outputs. Then using these sample states we can simulate coherent measurements in both  basis, and coherently erase the state of system.

Let $\{|\theta^\i_k\rangle: k=0\cdots d-1\}$ be $d$ unknown orthonormal states in the $d$-dimensional  \emph{input subspace} $\mathcal{H}_\i$, and  
\beq
|\alpha^\i_j\rangle=\frac{1}{\sqrt{d}}\sum^{d-1}_{k=0} e^{i2\pi k j/d }\ |\theta^\i_k\rangle\ \ \ \ \  :\  j=0\cdots d-1\ ,
\eeq
 be the orthonormal basis for $\mathcal{H}_\i$, which is conjugate to  $\{|\theta^\i_k\rangle: k=0\cdots d-1\}$.  We assume we are given multiple copies of the input states $\{|\theta^\i_k\rangle: k=0\cdots d-1\}$ and  $\{|\alpha^\i_j\rangle : j=0\cdots d-1\}$, and their corresponding output states $\{|\theta^\o_k\rangle=U|\theta^\i_k\rangle: k=0\cdots d-1\}$ and  $\{|\alpha^\o_j\rangle=U|\alpha^\i_j\rangle : j=0\cdots d-1\}$. Note that, even if we are only given multiple copies of states $\{|\theta^\i_k\rangle:\  k=0\cdots d-1\}$ and multiple copies of  \emph{one} of the states in  $\{|\alpha^\i_k\rangle: k=0\cdots d-1\}$, we can efficiently generate all  states in this set (and similarly, for the set $\{|\alpha^\o_k\rangle: k=0\cdots d-1\}$). \footnote{Let $P^\i=\sum_{k=0}^{d-1} k |\theta^\i_k\rangle\langle \theta^\i_k| $. Having multiple copies of states $\{|\theta^\i_k\rangle: k=0\cdots d-1\}$ we can simulate the unitary $\exp(-i P^\i s)$ for any real $s$, and by applying this unitary, for different values of $s$, we can transform  one element of the set   $\{|\alpha^\i_k\rangle: k=0\cdots d-1\}$ to the other elements.  To simulate the unitary $\exp(-i P^\i s)$, we note that  $\exp(-i P^\i s)= \prod_{k=0}^{d-1} \exp({- i s k |\theta^\i_k\rangle\langle \theta^\i_k|})$. As we have seen in  Appendix \ref{app:SWAP} each unitary $\exp({- i s k |\theta^\i_k\rangle\langle \theta^\i_k|})$ can be efficiently simulated using multiple copies of state $|\theta^\i_k\rangle$, in time $\mathcal{O}(\log(D))$. Therefore, we can simulate the unitary $\exp({-i P^\i s})$ in time $\mathcal{O}(d \log (D))$.  }

In the first step of this algorithm we perform a coherent measurement in $\{|\theta^\i_k\rangle: k=0\cdots d-1\}$ basis. To do this we couple the system to an ancillary system  with a $d-$dimensional Hilbert space.    The ancillary system  is initially prepared in the state $|\Gamma\rangle=\sqrt{d^{-1}}\sum_{t=0}^{d-1} |t\rangle$, where $\{|t\rangle: t=0\cdots d-1\}$ is a standard orthonormal basis.    Then, we use the given copies of sample states $\{|\theta^\i_k\rangle: k=0\cdots d-1\}$  to simulate the unitary 
\beq
V^\i_P=\sum_{k,t=0}^{d-1}    e^{i t k2\pi/d }  |\theta^\i_k\rangle\langle \theta^\i_k|\otimes |t \rangle \langle t|= \sum_{t=0}^{d-1}    e^{i t P^\i 2\pi/d}\otimes |t \rangle \langle t|\ ,
\eeq
where $P^\i=\sum_{k=0}^{d-1} k |\theta^\i_k\rangle\langle \theta^\i_k| $. To efficiently simulate this unitary we first note that it has a decomposition as
\beq\label{decomp45}
V^\i_P=\prod_{k=0}^{d-1}\sum_{t=0}^{d-1}    e^{i t |\theta^\i_k\rangle\langle\theta^\i_k| 2\pi k/d}\otimes |t \rangle \langle t|\ .
\eeq
Therefore, to simulate $V^\i_P$ we can simulate the (commuting) unitaries $\sum_{t=0}^{d-1}    e^{i t |\theta^\i_k\rangle\langle\theta^\i_k| 2\pi k/d}\otimes |t \rangle \langle t|$, for $k=0,\cdots , d-1$. Each unitary  $\sum_{t=0}^{d-1}    e^{i t |\theta^\i_k\rangle\langle\theta^\i_k| 2\pi k/d}\otimes |t \rangle \langle t|$   can be efficiently simulated using the given copies of state $|\theta^\i_k\rangle$. Using the results presented in Appendix \ref{app:SWAP}, this simulation can be done in time $\mathcal{O}(\log D)$, where $D$ is the dimension of the Hilbert space.  Then, it follows from the decomposition of $V^\i_P$ given by   Eq.(\ref{decomp45}) that we can simulate $V^\i_P$  in time $\widetilde{\mathcal{O}}(d \log D)$, using the given  copies of sample states  $\{|\theta^\i_k\rangle\ :  k=0\cdots d-1\}$.

Applying the unitary $V^\i_P$ to the system in state  $|\psi\rangle\in \mathcal{H}_\i$ and the ancillary system prepared in state  $|\Gamma\rangle$, we get state
\beq
|\psi\rangle\otimes |\Gamma\rangle=\sum_{k=0}^{d-1} \psi_k |\theta^\i_k\rangle\otimes |\Gamma\rangle  \longrightarrow V^\i_P\big(|\psi\rangle\otimes |\Gamma\rangle\big)= \sum_{k=0}^{d-1} \psi_k\  |\theta^\i_k\rangle \otimes \frac{1}{\sqrt{d}}\sum_{t=0}^{d-1} e^{i t k2\pi/d }|t\rangle\ ,
\eeq
where we have used the decomposition $|\psi\rangle=\sum_{k=0}^{d-1} \psi_k |\theta^\i_k\rangle$.

Then, performing  quantum Fourier transform on the ancillary system we transform the joint state to 
\beq\label{st345}
V^\i_P\big(|\psi\rangle\otimes |\Gamma\rangle\big)= \sum_{k=0}^{d-1} \psi_k\  \big(|\theta^\i_k\rangle \otimes \sum_{t=0}^{d-1} e^{i t k 2\pi/d}|k\rangle\big) \xrightarrow{\text{QFT}} \sum_{k=0}^{d-1} \psi_k\  |\theta^\i_k\rangle \otimes |k\rangle \ .
\eeq
Then, to erase the information in the system we implement the unitary    
\beq
V^\i_Q=\sum_{k,t=0}^{d-1}    e^{i t k2\pi/d }  |\alpha^\i_k\rangle\langle \alpha^\i_k|\otimes |t \rangle \langle t|= \sum_{t=0}^{d-1}    e^{i t Q^\i 2\pi/d}\otimes |t \rangle \langle t|\ ,
\eeq
where $Q^\i=\sum_{l=0}^{d-1} l |\alpha^\i_l\rangle\langle \alpha^\i_l| $. Note that we can efficiently simulate $V^\i_Q$, using a similar approach we used to simulate $V^\i_P$.

Applying $V^\i_Q$ to state in Eq.(\ref{st345}) we find
\beq
\sum_{k=0}^{d-1} \psi_k\  |\theta^\i_k\rangle \otimes |k\rangle
   \longrightarrow V^\i_Q   \sum_{k=0}^{d-1} \psi_k\  |\theta^\i_k\rangle \otimes |k\rangle   =  \sum_{k=0}^{d-1} \psi_k\  e^{i k Q^\i 2\pi/d}  |\theta^\i_k\rangle \otimes |k\rangle = |\theta^\i_0\rangle \otimes  \sum_{k=0}^{d-1} \psi_k\    |k\rangle\ ,
\eeq
where we have used the fact that 
\beq
e^{i k Q^\i 2\pi/d}  |\theta^\i_k\rangle=\sum_{l=0}^{d-1}  e^{i 2\pi lk/d}   |\alpha^\i_l\rangle\langle \alpha^\i_l  |\theta^\i_k\rangle=\frac{1}{\sqrt{d}}\sum_{l=0}^{d-1}  e^{i 2\pi lk/d} e^{-i 2\pi k l/d}  |\alpha^\i_l\rangle = |\theta^\i_{0}\rangle .
\eeq
Therefore, after these three steps  for any input state $|\psi\rangle\in \mathcal{H}_\i$ we have
\beq
\left[V^\i_Q(I\otimes F_\text{a})V^\i_P\right]  |\psi\rangle|\Gamma\rangle = |\theta^\i_0\rangle \otimes  \sum_{k=0}^{d-1} \psi_k\    |k\rangle \ ,
\eeq
where $F_\text{a}$ denotes the quantum Fourier transform on the ancillary system. 

At this point we have completely erased the state of system and transferred  all its information to the ancillary system. Now using a method similar to the one used in the main algorithm, we can transform this state to state $U|\psi\rangle|\Gamma\rangle$: 
We replace $|\theta^\i_0\rangle$ with $|\theta^\o_0\rangle$ and apply ${V^\o_P}^\dag (I\otimes F^\dag_\text{a}){V^\o_Q}^\dag$  to state $|\theta^\o_0\rangle \otimes  \sum_{k=0}^{d-1} \psi_k\    |k\rangle$, where
\begin{align}
V^\o_Q&=\sum_{k,t=0}^{d-1}    e^{i t k2\pi/d }  |\alpha^\o_k\rangle\langle \alpha^\o_k|\otimes |t \rangle \langle t|= \sum_{t=0}^{d-1}    e^{i t Q^\o 2\pi/d}\otimes |t \rangle \langle t|\ ,\\ 
V^\o_P&=\sum_{k,t=0}^{d-1}    e^{i t k2\pi/d }  |\theta^\o_k\rangle\langle \theta^\o_k|\otimes |t \rangle \langle t|= \sum_{t=0}^{d-1}    e^{i t P^\o 2\pi/d}\otimes |t \rangle \langle t|\ .
\end{align}
Note that these unitaries can be efficiently simulated using the same method we used to simulate the unitary $V_P^\i$.

It follows that we can emulate the action of unitary $U$ on the input subspace in time $\widetilde{\mathcal{O}}(d\log(D))$.

\newpage
\section{Performing Projective Measurements with Unknown Projectors}\label{App:meas}
In this section, we consider the problem of emulating projective  measurements. Suppose we are given multiple copies of states that belong to different subspaces of the Hilbert space, with the labels specifying these subspaces. Then, we are interested to simulate the projective measurement that projects any given input state to one of these subspaces. Note that any projective measurement can be realized as  a sequence of two-outcome measurements. Therefore, in the following, we focus on implementing two-outcome measurements. 

Consider the sample states $\{|\phi_k\rangle : 1,\cdots,K \}$.  
In contrast to the case of emulating unitaries, discussed in the main body of the paper, here  we do not make any assumptions about the sample states, or the relation between them. Let  $\mathcal{H}_\Pi$ be the subspace spanned by these states, and $\Pi$ be the projector to this subspace. We are interested in implementing the projective measurement described by the projectors $\{\Pi, I-\Pi\}$, assuming we are given multiple copies of each state in this set.  

In the algorithm we use the controlled-reflections
\beq
R_a(k)=|0\rangle\langle 0|_a\otimes I+|1\rangle\langle 1|_a\otimes e^{i\pi |\phi_{k}\rangle \langle\phi_{k}| }\ .
\eeq
As we explained in Appendix \ref{app:SWAP}, using $n$ copies of state $|\phi_k\rangle$ we can implement this unitary  in time $\mathcal{O}\big(\log(D) \times n\big)$,  with error  $\mathcal{O}(\frac{1}{n})$ in trace distance, where $D$ is the dimension of the Hilbert space.

\begin{figure*}
  \includegraphics[width=\textwidth,height=6.5cm]{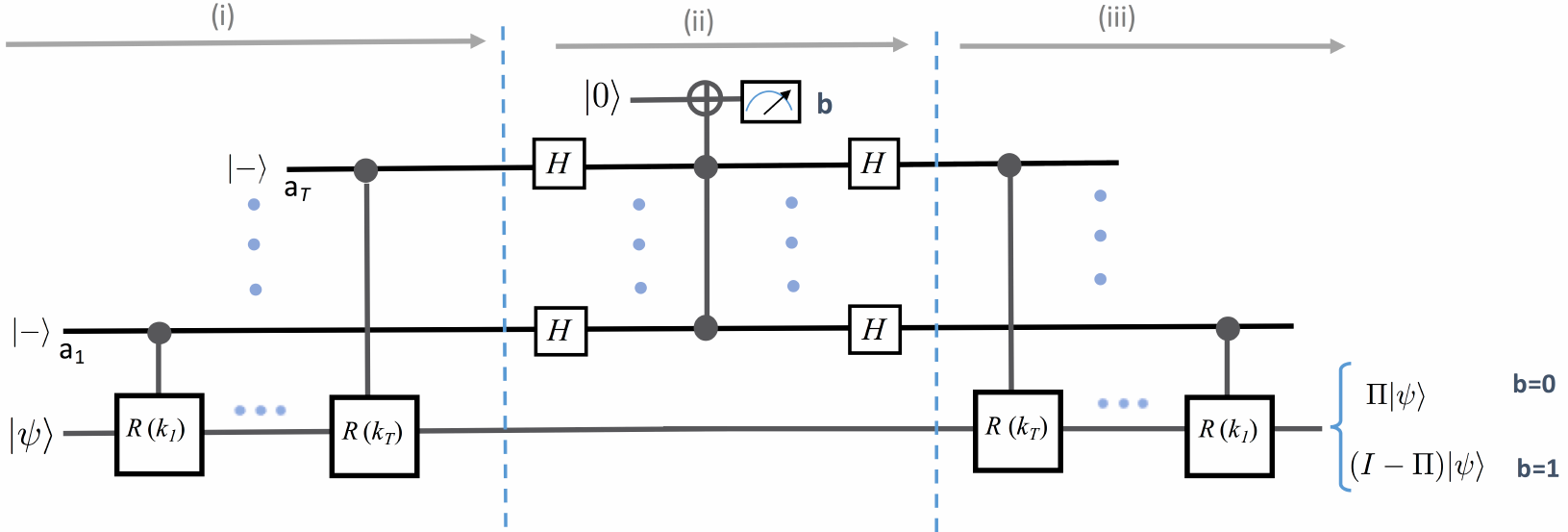}
\caption{The circuit for performing the two-outcome projective measurement with projectors  $\{\Pi, I-\Pi\}$, where $\Pi$ is the projector to the subspace spanned by    the sample states $\{|\phi_k\rangle : 1,\cdots,K \}$.
   \label{Fig3}}
\end{figure*}

\subsection{The Algorithm}

The quantum circuit for this algorithm is presented in Fig.(\ref{Fig3}). The algorithm has the following three steps:\\

\begin{itemize}
\item \noindent\textbf{(i)}  Consider $T$ qubits $a_1\cdots a_T$, which are all initially prepared in state $|-\rangle=(|0\rangle-|1\rangle)/\sqrt{2}$. Let $k_1,\cdots,k_T$ be $T$ independent random integers, each  chosen uniformly at random form $1,\cdots,K$. Apply the controlled-reflection $R_{a_1}(k_1)$ to the system and qubit $a_1$,  $R_{a_2}(k_2)$ to the system and qubit $a_2$, until the last controlled-reflection  $R_{a_T}(k_T)$,  which is applied to the system and qubit $a_T$.

\item \noindent\textbf{(ii)}  Perform the two-outcome projective measurement  $\{|-\rangle\langle -|^{\otimes T}, I-|-\rangle\langle -|^{\otimes T}\}$ on qubits $a_1\cdots a_T$.  This can be implemented efficiently, e.g., by first applying the Hadamard gates on qubits $a_1\cdots a_T$, then applying a  $T$-bit Toffoli gate  controlled by $a_1\cdots a_T$, that is acting on an ancillary qubit initialized in state $|0\rangle$, and  finally applying the Hadamard gates on qubits $a_1\cdots a_T$ again. Then, measuring this ancillary qubit in the computational basis, realizes the desired measurement.

\item\noindent\textbf{(iii)} Recall the random integers  $k_1,\cdots,k_T$ chosen in step (i).  Apply the controlled-reflection $R_{a_T}(k_T)$ to the system and qubit $a_T$, then apply  the controlled-reflection $R_{a_{T-1}}(k_{T-1})$ to the system and qubit $a_{T-1}$, until the last controlled-reflection $R_{a_1}(k_{1})$ which is applied to the system and qubit $a_1$.

\end{itemize}

As we show in the following,    for sufficiently large $T$, the initial state $|\psi\rangle$ is projected to a state close to $\Pi|\psi\rangle/\sqrt{\langle\psi|\Pi|\psi\rangle}$ with probability close to  $\langle\psi|\Pi|\psi\rangle$ (corresponding to outcome $b=0$ in the measurement in step (ii)), and  is projected to a state close to $(I-\Pi)|\psi\rangle/\sqrt{1-\langle\psi|\Pi|\psi\rangle}$ with probability  close to $1-\langle\psi|\Pi|\psi\rangle$ (corresponding to outcome $b=1$ in the measurement).        

\subsection{How it works}

In the following, we  assume the circuit  is implemented with perfect controlled-reflections. Then, the additional error due to the imperfect implementation  of the controlled-reflections can be taken into account, using the same approach we used before for the main algorithm.

Let $\mathcal{E}$ be the quantum channel corresponding to  the overall effect of the circuit in Fig.(\ref{Fig3}) in the case where we ignore the output of measurement in step (ii), i.e., we do not post-select. Let $\mathcal{H}_\Pi$ be the subspace spanned by the sample states $\{|\phi_k\rangle\ :k=1,\cdots,K\}$. 

\begin{theorem}\label{lastThm}
Let $\text{Pr}(b=0)$ be the probability that the measurement in step (ii) of the circuit in Fig.(\ref{Fig3}) returns outcome 0. 
Then,
\be
\Tr(\Pi \rho)-(1- \lambda_\text{min})^T \le \text{Pr}(b=0)\le \Tr(\Pi \rho)  \ ,
\ee
where  $\lambda_\text{min}>0$ is the minimum non-zero eigenvalue of the density operator 
\beq
\sigma_{\text{avg}}\equiv \frac{1}{K} \sum_{k=1}^K |\phi_k\rangle\langle\phi_k|\ .
\ee
 Furthermore, if $\rho$ is block-diagonal with respect to the subspace  $\mathcal{H}_\Pi=\text{Span}_\mathbb{C}\{|\phi_k\rangle\ :k=1,\cdots,K\}$ and the orthogonal subspace,  such that $[\Pi, \rho]=0$,  then the  fidelity of $\mathcal{E}(\rho)$, the output of the circuit,   with $\rho$  the desired output state, is lower bounded by 
\be
F(\mathcal{E}(\rho),\rho)\ge 1-(1- \lambda_\text{min})^T\ .
\ee
\end{theorem}
This theorem implies that in the limit  $T\rightarrow\infty$, the probability of outcome $b=0$ converges to $\Tr(\Pi \rho)$. Furthermore, if $\rho$ is block-diagonal with respect to $\mathcal{H}_\Pi$, it remains unchanged under the measurement. Therefore, the circuit realizes the perfect projective measurement with projectors $\{\Pi, I-\Pi\}$.\\

\begin{proof} In the following, we determine the probability of outcome $b=1$ in step (ii), which is equal to the probability that at the end of step (i) the ancilla qubits $a_1,\cdots, a_T$ are in state $|-\rangle^{\otimes T}$. That is, 
\be
\text{Pr}(b=1)=\frac{1}{K^T}\sum_{k_1,\cdots,k_T=1}^K\Tr\Big([|-\rangle\langle-|^{\otimes T}\otimes I] \ [R_{a_T}(k_T) \cdots R_{a_1}(k_1)]\ \left[ |-\rangle\langle-|^{\otimes T}\otimes \rho\right]\  [R_{a_1}(k_1)\cdots  R_{a_T}(k_T)] \Big) \ .
\ee
Then, using the fact that
\beq
\langle-| R_a(k)|-\rangle_a=\frac{1}{2}(I+e^{i\pi |\phi_k\rangle\langle\phi_k|})=I-|\phi_k\rangle\langle\phi_k|\ ,
\eeq
we find 
\bes\label{meas1}
\begin{align}
\text{Pr}(b=1) &=
\frac{1}{K^T}\sum_{k_1,\cdots,k_T=1}^K \Tr\Big([I-|\phi_{k_T}\rangle\langle\phi_{k_T}|]\cdots [I-|\phi_{k_1}\rangle\langle\phi_{k_1}|]\ \rho\ [I-|\phi_{k_1}\rangle\langle\phi_{k_1}|]\cdots [I-|\phi_{k_T}\rangle\langle\phi_{k_T}|] \Big)\\ &=\Tr\left(\mathcal{M}^T(\rho)\right)\  ,
\end{align}
\ees
where 
\be
\mathcal{M}(\cdot)=\frac{1}{K} \sum_{k=1}^K P_k^\perp(\cdot)P_k^\perp\ ,
\ee
 and $P_k^\perp=I-|\phi_{k}\rangle\langle\phi_{k}|$. Note that $\mathcal{M}(\rho)$ has a simple interpretation: We perform one of $K$ projective measurements $\{|\phi_k\rangle\langle\phi_k|, I-|\phi_k\rangle\langle\phi_k\}$ with equal probability. Then, $\mathcal{M}(\rho)$ is the unnormalized state assuming we have not observed the outcome corresponding to rank-1 projectors $\{|\phi_k\rangle\langle\phi_k|\}$.

Then,
 \begin{align}
\text{Pr}(b=1)&=\Tr\left(\mathcal{M}^T(\rho)\right)\\ &= \Tr\left(\Pi\mathcal{M}^T(\rho)\Pi\right)+\Tr\left(\Pi_\perp\mathcal{M}^T(\rho)\Pi_\perp\right)\\ &=  
\Tr\left(\mathcal{M}^T(\Pi\rho \Pi)\right)+\Tr\left(\mathcal{M}^T(\Pi_\perp\rho \Pi_\perp)\right)\ ,\label{kek}
\end{align}
where $\Pi_\perp=I-\Pi$ is the projector to the orthogonal complement of $\mathcal{H}_\Pi$, and we have used the fact that operators $P_k^\perp$ commute with $\Pi_\perp$ and $\Pi$. Furthermore, because $P_k^\perp$ acts trivially on the subspace orthogonal to $\mathcal{H}_\Pi$,   $\Pi_\perp P_k^\perp =\Pi_\perp$ for all $k$, it can 
be easily seen that
\be\label{kek2}
\Tr\left(\mathcal{M}^T(\Pi_\perp\rho \Pi_\perp)\right)=\Tr(\rho \Pi_\perp)\ .
\ee
This can also be seen more directly by noting that if the input state is in the subspace orthogonal to $\mathcal{H}_\Pi$, then all controlled-reflection acts trivially, which means  with probability 1 we obtain outcome $b=1$. Next, we focus on the term $\Tr\left(\mathcal{M}^T(\Pi\rho \Pi)\right)$.  

For any  operator $X$ we have
\beq\label{sdvv}
\Tr(\mathcal{M}(X))= \frac{1}{K} \sum_{k=1}^K\Tr(X P_k^\perp  )=\Tr(X)- \Tr\Big(X \sum_{k=1}^K \frac{1}{K} |\phi_k\rangle\langle\phi_k|\Big)=\Tr(X)- \Tr(X \sigma_{\text{avg}}) \  , 
\eeq
where  $\sigma_{\text{avg}}\equiv \frac{1}{K} \sum_{k=1}^K |\phi_k\rangle\langle\phi_k|$. The subspace $\mathcal{H}_\Pi$ is defined as the subspace spanned by $\{|\phi_k\rangle: k=1,\cdots,K\}$. Therefore, the density matrix $\sigma_{\text{avg}}$ is automatically full-rank in this subspace. Let $\lambda_\text{min}>0$ be the minimum eigenvalue of  $\sigma_{\text{avg}}$ in this subspace, i.e.,  the minimum nonzero eigenvalue of $\sigma_{\text{avg}}$. Note that for any positive semi-definite operator $X\ge 0$ whose support is restricted to $\mathcal{H}_\Pi$ we have
\beq
\Tr(X \sigma_{\text{avg}}) \ge \lambda_\text{min} \Tr(X)\ .
\eeq
Then, using Eq.(\ref{sdvv}) this means that for any $X\ge 0$, whose support is restricted to $\mathcal{H}_\Pi$ we have
\beq
\Tr(\mathcal{M}(X))=\Tr(X)- \Tr(X \sigma_{\text{avg}})\le \Tr(X) (1- \lambda_\text{min}) \  . 
\eeq
Next, note that if $X$ is a positive semi-definite operator with support restricted to $\mathcal{H}_\Pi$, then $\mathcal{M}(X)$ is also a positive  semi-definite  operator with support restricted to $\mathcal{H}_\Pi$. Therefore, we conclude that
\beq
\Tr(\mathcal{M}^T(X))\le \Tr(X) \times (1- \lambda_\text{min})^T \  . 
\eeq
Using this bound for $X=\Pi \rho \Pi$, and combining Eq.(\ref{kek}) and Eq.(\ref{kek2}) we conclude that 
\be
\Tr(\rho\Pi_\perp) \le \text{Pr}(b=1) =\Tr\left(\mathcal{M}^T(\rho)\right)=\Tr(\rho\Pi_\perp)+\Tr(\mathcal{M}^T(\Pi\rho\Pi))\le \Tr(\rho\Pi_\perp)+  (1- \lambda_\text{min})^T\ ,
\ee
which, in turn implies
\be
\Tr(\rho\Pi)-(1- \lambda_\text{min})^T  \le \text{Pr}(b=0)\le  \Tr(\rho\Pi) \ .
\ee
This completes the proof of the first part of the theorem.

Finally, the second part of the theorem, i.e., the bound on the fidelity,  follows from the following lemma, which is proven in the same way we proved Theorem \ref{Thm_fid}. 
\begin{lemma}
Consider  quantum operation 
\beq
\mathcal{G}(\cdot)=\sum_{\boldsymbol{\lambda}}  p(\boldsymbol{\lambda}) \sum_{\alpha\in A} V^\dag(\boldsymbol{\lambda}) 
P_\alpha V(\boldsymbol{\lambda}) (\cdot)V^\dag(\boldsymbol{\lambda})  P_\alpha  V(\boldsymbol{\lambda}) \ ,
\eeq
which describes the process in which we  apply a unitary $V(\boldsymbol{\lambda})$ to the system, where  $\boldsymbol{\lambda}$ is a random parameter chosen  according to the probability distribution $p(\boldsymbol{\lambda})$, then  we perform a projective measurement on the system with projectors $\{P_\alpha:\alpha\in A\}$, and then finally we apply the unitary $V^\dag(\boldsymbol{\lambda})$.  For any input state $\rho$, the fidelity of $\rho$ and $\mathcal{G}(\rho)$ is lower bounded by 
\beq
F(\rho,\mathcal{G}(\rho))\ge \underset{\alpha\in A}{\max}\    p_\alpha(\rho)\ ,
\eeq
where 
\be
p_\alpha(\rho)=\sum_{\boldsymbol{\lambda}}  p(\boldsymbol{\lambda}) \Tr(V(\boldsymbol{\lambda}) \rho V^\dag(\boldsymbol{\lambda})  P_\alpha)  
\ee
 is the (expected) probability of outcome $\alpha$ in the measurement.
\end{lemma}
This bound essentially means that if one of the outcomes has probability close to 1, then the overall process does not affect the state of system (which is in the same spirit of the Gentle measurement lemma   \cite{wilde2013quantum}).\\
 
\begin{proof}
Consider a pure input state $|\gamma\rangle$. Then, the squared fidelity of  $\mathcal{G}(|\gamma\rangle\langle\gamma|)$ and $|\gamma\rangle\langle\gamma|$ is given by
\begin{align}
F^2(\mathcal{G}(|\gamma\rangle\langle\gamma|),|\gamma\rangle\langle\gamma|)&=\langle\gamma|\mathcal{G}(|\gamma\rangle\langle\gamma|) |\gamma\rangle\\ &=  \sum_{\alpha\in A}  \sum_{\boldsymbol{\lambda}}  p(\boldsymbol{\lambda}) \langle\gamma|V^\dag(\boldsymbol{\lambda})  P_\alpha  V(\boldsymbol{\lambda}) |\gamma\rangle^2\\ &\ge  \underset{\alpha\in A}{\max}\  \sum_{\boldsymbol{\lambda}}  p(\boldsymbol{\lambda}) \langle\gamma|V^\dag(\boldsymbol{\lambda})  P_\alpha  V(\boldsymbol{\lambda}) |\gamma\rangle^2  \\  &\ge \underset{\alpha\in A}{\max}\  \Big(\sum_{\boldsymbol{\lambda}}  p(\boldsymbol{\lambda}) \langle\gamma|V^\dag(\boldsymbol{\lambda})  P_\alpha  V(\boldsymbol{\lambda}) |\gamma\rangle\Big)^2 \\ &= \underset{\alpha\in A}{\max}\  p_\alpha^2(|\gamma\rangle\langle\gamma|)\ ,
\end{align}
where to get the fourth line we have used the fact that the variance of any random variable is non-negative, and to get the last line we have used the fact that the average density operator of the system before the measurement is $\sum_{\boldsymbol{\lambda}}  p(\boldsymbol{\lambda}) V(\boldsymbol{\lambda}) |\gamma\rangle\langle\gamma|V^\dag(\boldsymbol{\lambda})  $, and so the probability of outcome $\alpha$ is  $\sum_{\boldsymbol{\lambda}}  p(\boldsymbol{\lambda}) \langle\gamma|V^\dag(\boldsymbol{\lambda})  P_\alpha  V(\boldsymbol{\lambda}) |\gamma\rangle$. This proves the lemma for the special case of pure states. The result for general mixed states follows from the joint concavity of fidelity.
\end{proof}

Now consider the two special cases of density operators
with support restricted to $\mathcal{H}_\Pi$ or its orthogonal subspace $\mathcal{H}^\perp_\Pi$, such that  $\Pi\rho_\text{inside}\Pi=\rho_\text{inside}$ and  
$\Pi_\perp\rho_\text{outside}\Pi_\perp=\rho_\text{outside}$. Then, applying the first part of the theorem, we find that if $\Pi\rho_\text{inside}\Pi=\rho_\text{inside}$ then $\text{Pr}(b=0)\ge 1- (1-\lambda_\text{min})^T$. In this case the above lemma implies
\be
F(\mathcal{E}(\rho_\text{inside}),\rho_\text{inside}) \ge 1- (1-\lambda_\text{min})^T\ . 
\ee
On the other hand, if $\Pi_\perp\rho_\text{outside}\Pi_\perp=\rho_\text{outside}$, then $\text{Pr}(b=1)=1$, in which case 
\be
F(\mathcal{E}(\rho_\text{outside}),\rho_\text{outside}) = 1\ . 
\ee 
Now  general $\rho$ that is  
 block-diagonal with respect to $\mathcal{H}_\Pi$, can be written as a convex combination of two density operators $\rho_\text{inside}$ and $\rho_\text{outside}$ satisfying the above constraints. Then, the joint concavity of fidelity implies that for any such $\rho$, $F(\mathcal{E}(\rho),\rho)\ge 1-(1- \lambda_\text{min})^T$. 
 This completes the proof of theorem \ref{lastThm}.
\end{proof}
\end{document}